\documentclass[onecolumn,aps,nofootinbib,superscriptaddress,tightenlines,notitlepage,11pt,pra]{revtex4-1}
\usepackage{newfloat}
\usepackage{lipsum}  
\usepackage{latexsym}
\usepackage{dcolumn}
\usepackage{graphicx}
\usepackage{amssymb}
\usepackage{amsmath}
\usepackage{amsfonts}
\usepackage{mathrsfs}
\usepackage{float}
\usepackage{color}
\usepackage{algorithm}
\usepackage{algpseudocode}
\usepackage{xcolor}
\usepackage{courier}
\usepackage{enumitem}
\usepackage{bbold}
\usepackage{mathtools}
\usepackage[breaklinks,colorlinks = true,linkcolor = blue,urlcolor  = blue,citecolor = red,anchorcolor = blue,]{hyperref}
\usepackage{appendix}

\newcommand{\bra}[1]{\langle #1|}

\newcommand{\ket}[1]{|#1\rangle}

\newcommand{\ketbra}[2]{|#1\rangle\!\langle #2|}
\newcommand{\vertiii}[1]{{\left\vert\kern-0.25ex\left\vert\kern-0.25ex\left\vert #1 \right\vert\kern-0.25ex\right\vert\kern-0.25ex\right\vert}}
\newcommand{\Tr}[1]{\text{Tr}\left(#1\right)}
\newcommand{\qedsymbol}{\rule{0.5em}{0.5em}}

\def\cB{\mathcal B}

\def\cD{\mathcal D}
\def\cE{\mathcal E}
\def\cF{\mathcal F}
\def\cG{\mathcal G}

\def\cI{\mathcal I}

\def\cL{\mathcal L}
\def\cM{\mathcal M}
\def\cN{\mathcal N}
\def\cP{\mathcal P}

\def\cS{\mathcal S}
\def\cT{\mathcal T}

\def\cX{\mathcal X}

\DeclareMathOperator*{\argmax}{arg\,max}

\newtheorem{theorem}{Theorem}[section]
\newtheorem{corollary}[theorem]{Corollary}

\newtheorem{prop}[theorem]{Proposition}

\newtheorem{lemma}[theorem]{Lemma}

\newtheorem{remark}[theorem]{Remark}
\newenvironment{proof}[1][Proof]{\noindent\textbf{#1.} }{\ \rule{0.5em}{0.5em}}
\newtheorem{definition}[theorem]{Definition}

\allowdisplaybreaks

\begin{document}

\title{Moderate deviation expansion for fully quantum tasks}

\author{Navneeth Ramakrishnan}
\altaffiliation{Part of this work appeared in the Information Theory Workshop 2021 \cite{9457638}.}
\affiliation{Department of Computing, Imperial College London, London, UK}
\affiliation{Center for Quantum Technologies, National University of Singapore, Singapore}

\author{Marco Tomamichel}
\affiliation{Center for Quantum Technologies, National University of Singapore, Singapore}
\affiliation{Department of Electrical and Computer Engineering, National University of Singapore, Singapore}

\author{Mario Berta}
\altaffiliation{This work was completed prior to MB joining the AWS Center for Quantum Computing.}
\affiliation{Department of Computing, Imperial College London, London, UK}
\affiliation{Institute for Quantum Information and Matter, California Institute of Technology, Pasadena, USA}
\affiliation{AWS Center for Quantum Computing, Pasadena, USA}


\begin{abstract}
The moderate deviation regime is concerned with the finite block length trade-off between communication cost and error for information processing tasks in the asymptotic regime, where the communication cost approaches a capacity-like quantity and the error vanishes at the same time. We find exact characterisations of these trade-offs for a variety of fully quantum communication tasks, including quantum source coding, quantum state splitting, entanglement-assisted quantum channel coding, and entanglement-assisted quantum channel simulation. The main technical tool we derive is a tight relation between the partially smoothed max-information and the hypothesis testing relative entropy. This allows us to obtain the expansion of the partially smoothed max-information for i.i.d.\ states in the moderate deviation regime.
\end{abstract}

\maketitle


\section{Overview}

Fundamental tasks in information theory are often characterized in the asymptotic setting: One considers i.i.d.\ copies of a resource such as a channel or a state, takes the limit of the number of copies to infinity, and then enforces that the error vanishes in the asymptotic limit to obtain an optimal coding rate~\cite{shannon1948mathematical}. While such results are fundamental theoretical bounds, practical settings involve only finitely many copies of the resource\,---\,the so-called finite block length setting. It is therefore of interest to also understand the quantitative trade-off between the error and the rate for a given block length.

In the moderate deviation setting one obtains a rate and an error in terms of the block length $n$, such that the rate approaches a capacity-like quantity and the error vanishes as the code length $n$ increases. In contrast, second-order (small deviation) analyses~\cite{strassen1962asymptotische, li2014second, tomamichel2013hierarchy, hayashi2009information, polyanskiy2010channel} result in a rate that approaches the capacity but the error is only bounded by some constant. As for error exponent (large deviation) analyses~\cite{csiszar1971error, han1989strong, arimoto1973converse, dueck1979reliability, mosonyi2017strong}, these result in an error that vanishes exponentially with $n$ but the difference between the rate and the capacity is a constant term. The moderate deviation regime is natural to consider since it achieves both these goals simultaneously and several classical information-theoretic tasks have been analysed in this setting. Classical channel coding has been studied in~\cite{altuug2014moderate, polyanskiy2010channel}, while asymmetric binary hypothesis testing has been studied in~\cite{sason2012moderate}. Moderate deviation analyses for several classical-quantum coding tasks and quantum hypothesis testing are found in~\cite{rouze2017finite, cheng2017moderate, chubb2017moderate, cheng2020non}. Other settings such as majorization-based resource interconversion have also been explored \cite{chubb2019moderate}. Recently, moderate deviation analyses have been done for private communication over wiretap channels and entropy accumulation~\cite{shen2022strong}, quantum soft covering~\cite{shen2022optimal, cheng2022error}, and privacy amplification against quantum side information~\cite{shen2022strong, shen2022optimal}.

A moderate sequence is defined as a sequence of positive numbers $\{a_n\}$ such that $\lim_{n\rightarrow\infty}a_n = 0$ and $\lim_{n\rightarrow\infty}a_n\sqrt{n} = \infty$~\cite{chubb2017moderate}. We furthermore introduce the strict moderate deviation regime, where we impose additionally that the moderate sequence satisfies $\lim_{n\rightarrow\infty}(na_n^2)^{-1}\log n = 0$. An example of a strictly moderate sequence is $\{n^{-\alpha}\}$ for $n\in\mathbb{N}$ and $\alpha\in\left(0,\frac{1}{2}\right)$ and in fact, this definition has been used in classical works to define a moderate sequence~\cite{hayashi2016uniform, watanabe2017finite, hayashi2020finite}. Within the moderate deviation regime, one can consider the low-error case where the error in the protocol is chosen to be $\varepsilon_n = e^{-a_n^2n}$ and the high-error case where the error is chosen to be $\varepsilon_n = 1 - e^{-a_n^2n}$.\footnote{We resolve either the high-error or the low-error for each quantum information-theoretic task we consider.}

As an illustrative example of our work, consider the task of quantum source coding: Alice holds a quantum state $\rho_A^{\otimes n}$ and the goal is to send this state to Bob using a noiseless channel while preserving correlations with a purifying reference quantum register. A fundamental result in quantum information theory~\cite{schumacher1995quantum} states that the asymptotic rate at which this can be achieved is given by the von Neumann entropy $S(A)_\rho$. We consider the setting for finite $n\in\mathbb N$ and aim to achieve quantum source coding such that the output of the protocol $\tilde{\rho}_{B^n}$ is $\varepsilon_n$-close in purified distance to the ideal output state $\rho_B^{\otimes n}=\mathcal{I}_{A^n\to B^n}\left(\rho_A^{\otimes n}\right)$. We show in Theorem~\ref{thm:source_coding_ea_low_error} that the minimal quantum communication cost $q^\star_{\varepsilon_n}(\rho_A^{\otimes n})$ of $\varepsilon_n$-error quantum source coding of $\rho_{A}^{\otimes n}$ with $\varepsilon_n = e^{-na^2_n}$ for a moderate sequence $\{a_n\}$ is given by 
\begin{align}
\frac{1}{n}q^\star_{\varepsilon_n}(\rho_{A}^{\otimes n}) =S(A)_\rho +  2a_n\sqrt{V(A)_\rho} + o(a_n)\,,
\end{align}
where $V(A)_\rho = \text{Tr}\left(\rho(\log\rho)^2\right) - S(A)_\rho^2$ denotes the quantum varentropy and $o(a_n)=f(n)$ for functions with $\lim_{n\rightarrow\infty}f(n)a_n^{-1} = 0$.

We show similar moderate deviations analyses for quantum state splitting~\cite{Abeyesinghe09, anshu2017quantum} (Theorem~\ref{thm:state_splitting_mod_devation}), entanglement-assisted quantum channel coding~\cite{bennett2002entanglement} (Theorem~\ref{thm:channel_coding_mod_dev}), and quantum channel simulation~\cite{bennett2002entanglement, bennett2009quantum, berta2011quantum} (Theorem~\ref{thm:channel_sim_mod_dev_final_result}). For quantum source coding, both second-order~\cite{abdelhadi2020second} and error exponent~\cite{hayashi2002exponents, hayashi2009information} analyses already exist and it is possible to derive our moderate deviation results from these results. Similarly, second-order analysis~\cite{datta2016second} for entanglement-assisted channel coding may also be used to derive our moderate deviation results.

Some of our technical discussion may be of independent interest: First, we use a fidelity-based channel distinguishability measure\,---\,the channel purified distance (Definition \ref{def:purified_diamond_distance_original})\,---\,which emerges as a natural fully quantum measure to analyse channel coding and channel simulation problems (versus trace-distance-based measures such as the diamond distance). Second, we determine when the tight triangle inequality for the purified distance \cite[Proposition 3.16]{tomamichel2015quantum} holds (Lemma~\ref{lem:purified_dist_triangle_ineq}). Third, from this refined triangle inequality, we derive a tight entropy inequality between the min-relative entropy and the hypothesis testing relative entropy (Lemma~\ref{lem:dh_dmin_bound}). At the same time, we then find the expansion of the partially smoothed max-information for i.i.d.\ states in the moderate deviation regime.

Our paper is organized as follows. We introduce the relevant notation in Section~\ref{sec:notation}, followed by our analysis of quantum state splitting in Section~\ref{sec:state split}\,---\,where we also introduce the necessary mathematical moderate deviation tools. Quantum source coding is discussed in Section~\ref{sec:extensions}, quantum channel simulation in Section~\ref{sec:channel_sim_mod_dev}, and quantum channel coding in Section~\ref{sec:channel_coding}. We conclude and discuss some open problems in Section~\ref{sec:conclusion}. Various technical proofs are deferred to Appendices \ref{app:state_splitting_proofs}--\ref{app:meta_converse_comparison}. 

\textit{Note added:} During finalization of our work we became aware of the related concurrent work~\cite{li2021reliable} by Li and Yao.


\section{Notation}\label{sec:notation}

\subsection{Mathematics}

The exponential and logarithm functions, denoted by $\exp$ and $\log$ respectively, are taken with respect to base 2 unless otherwise specified. The set of real numbers is denoted by $\mathbb{R}$ and the set of natural numbers by $\mathbb{N}$. We use the big-O and little-o notation which are defined as follows: For functions $f(n), g(n)$ on $\mathbb{N}$, we denote $f(n) = O(g(n))$ if there exist some positive constants $c$ and $n_0$ such that for all $n\geq n_0$, it holds that $|f(n)|\leq c |g(n)|$. We denote $f(n) = o(g(n))$ if $\lim_{n\rightarrow\infty}\frac{f(n)}{g(n)} = 0$.


\subsection{Quantum states}

We associate quantum registers with Hilbert spaces and assume them to be finite-dimensional. We label quantum registers and the associated Hilbert spaces with capital letters, e.g., $A$, $B$, and denote their dimension by $|A|$, $|B|$, and so on. The notation $A\cong B$ states that $A$ and $B$ are isomorphic. The tensor product of spaces $A$ and $B$ is denoted by $AB=A\otimes B$. The set of linear operators on $A$ is denoted by $\cL(A)$ and the set of positive semi-definite operators on $A$ is denoted by $\cP(A)$, where for elements $\rho_A\in\cP(A)$ we write $\rho_A\succeq 0$. The set of positive semi-definite operators with $\Tr{\rho_A}\leq 1$\,---\,called sub-normalized quantum states\,---\, is denoted by $\cS_{\leq}(A)$ and the set of positive semi-definite operators with unit trace\,---\,called normalized quantum states or quantum states for short\,---\,is denoted by $\cS(A)$. If the rank of $\rho_A\in\cS_\leq(A)$ is one, then the (sub-normalized) quantum state is called pure. An extension of $\rho_{A}\in\cS_{\leq}(A)$ is $\rho_{AB}\in\cS_{\leq}(AB)$ such that $\text{Tr}_B(\rho_{AB}) = \rho_A$, where $\text{Tr}_B$ denotes taking a trace over the quantum register $B$ (see~\cite[Definition 4.3.4]{wilde2011classical} for a definition of the partial trace). If an extension $\rho_{AB}$ of $\rho_A\in\cS_\leq$ is pure, then it is called a purification.

When the registers associated with a state are clear from context or do not need to be explicitly stated, we drop them for simplicity of notation. For $\rho, \sigma\succeq0$ with the support of $\rho$ contained in the support of $\sigma$, we write $\rho\ll \sigma$. Orthogonality for $\rho,\sigma\succeq0$ is defined as $\rho\sigma=\sigma\rho = 0$ and denoted as $\rho\perp\sigma$. The identity operator on a $d$-dimensional quantum register is denoted by $I_d$. The generalized inverse~\cite{penrose1955generalized} of $\rho\in\cS_{\leq}(A)$ is denoted by $\rho^{-1}$ and is the inverse on the support of the quantum state. 


\subsection{Quantum operations}

A quantum channel $\cN_{A \rightarrow B}$ is a linear completely positive trace-preserving (CPTP) map taking $\cL(A)$ to $\cL(B)$. The identity quantum channel between isomorphic quantum registers $A$ and $B$ is denoted by $\cI_{A\rightarrow B}$ and we abbreviate $\cI_{A\rightarrow A}=\cI_A$. We denote composition of two quantum channels $\cN_{A\rightarrow B}$ and $\cM_{B\rightarrow C}$ by $\cM_{B\to C}\circ \cN_{A\to B}$.

A positive operator valued measure (POVM) is a finite set $\{M_i\}_{i=1}^n$ of $d$-dimensional positive definite matrices $M_i\succeq0$ such that $\sum_{i=1}^n M_i = I_d$ for some $d, n\in\mathbb{N}$.


\subsection{Distance measures}

The trace norm of a linear operator $M$ is given by $\|M\|_1 = \mathrm{Tr}\left(\sqrt{M^\dagger M}\right)$, where $M^\dagger$ is the transpose conjugate of $M$. The purified distance between $\rho,\sigma \in \cS_{\leq}(A)$ is given by $P(\rho, \sigma)=\sqrt{1-\bar{F}(\rho, \sigma)^2}$, where $\bar{F}(\rho,\sigma)$ is the generalized fidelity defined as \cite{tomamichel2012framework}
\begin{align}
\bar{F}(\rho, \sigma)=\|\sqrt{\rho \oplus(1-\Tr{\rho})} \sqrt{\sigma \oplus(1-\Tr{\sigma})}\|_{1}\,.
\end{align}

Note that if either $\rho$ or $\sigma$ have unit trace, then the generalized fidelity is equal to the standard fidelity $F(\rho,\sigma) = \|\sqrt{\rho} \sqrt{\sigma}\|_{1}$. For $\rho,\sigma\in\cS_{\leq}(A)$ with $P(\rho,\sigma)\leq\varepsilon$ we write $\rho\approx_{\varepsilon}\sigma$. We define the $\varepsilon$-ball around $\rho\in\cS_{\leq}(A)$ by $\mathcal{B}^{\varepsilon}(\rho)=\left\{\bar{\rho} \in \mathcal{S}_{\leq}(A): \rho\approx_\varepsilon\bar{\rho} \leq \varepsilon\right\}$. For $\rho,\sigma\in\cS(A)$ the purified distance is bounded by the trace distance as \cite{fuchs1999cryptographic}
\begin{align}
\sqrt{2\|\rho - \sigma\|_1} \geq P(\rho,\sigma) \geq \|\rho - \sigma\|_1\,.
\end{align}


\subsection{Entropic quantities}

For $\rho\in\cS(A)$, the von Neumann entropy is defined as
\begin{align}
S(A)_\rho = S(\rho)=- \Tr{\rho \log \rho}\,.
\end{align}
For $\rho\in\cS_{\leq}(A)$ and $\sigma\in\cP(A)$ the quantum relative entropy is defined as
\begin{align}
D(\rho \| \sigma)= \Tr{\rho \left(\log \rho -\log \sigma \right)}
\end{align}
when $\rho\ll\sigma$ and $\infty$ otherwise. The relative entropy variance is given by \cite{tomamichel2013hierarchy, li2014second}
\begin{align}
V(\rho\|\sigma) = \Tr{\rho(\log\rho - \log\sigma)^2} - D(\rho\|\sigma)^2
\end{align}
when $\rho\ll\sigma$ and $\infty$ otherwise. Following the classical definition \cite{kontoyiannis2013optimal}, we define the quantum varentropy of $\rho\in\cS(A)$ to be
\begin{align}
V(A)_\rho = \Tr{\rho(\log\rho)^2} - S(A)_\rho^2\,.
\end{align}
Furthermore, we define the channel based entropic quantities for quantum channels $\cN_{A\rightarrow B}$ as
\begin{align}\label{eq:channel-capacities}
C(\cN) = \max\limits_{\sigma_{AR}\in\cS(AR)} \frac{1}{2}I(B:R)_{(\cN\otimes \cI)(\sigma)}\quad\text{and}\quad V_{\max}(B:R)_{\cN} = \max\limits_{\sigma'\in\Pi(\cN)}V(B:R)_{(\cN\otimes \cI)(\sigma')}\, ,
\end{align}
where $R$ is any purifying quantum register of the input of the channel and
\begin{align}
\label{eq:capacity_achieving_input_set}
\Pi(\cN) = \argmax\limits_{\tau} I(B:R)_{(\cN\otimes \cI)(\tau)}
\end{align}
is called the set of capacity achieving inputs for $\cN$.

The sandwiched R\'enyi relative entropies for $\rho\in\cS(A)$ and $\sigma\in\cP(A)$ are defined~\cite{muller2013quantum, Wilde14} as follows: For $\alpha \in(0,1) \cup(1, \infty)$, if either $\alpha<1$ and $\rho \not \perp \sigma$ or $\alpha>1$ and $\rho \ll \sigma$, we have 
\begin{align}
\widetilde{D}_{\alpha}(\rho \| \sigma)=\frac{1}{\alpha-1} \log \operatorname{Tr}\left(\left(\sigma^{\frac{1-\alpha}{2 \alpha}} \rho \sigma^{\frac{1-\alpha}{2 \alpha}}\right)^{\alpha}\right)\,.
\end{align} 
Furthermore, for $\alpha = \{0,1,\infty\}$ the respective limits are employed. The max-relative entropy of $\rho\in\cS_{\leq}(A)$ with respect to $\sigma\in\cP(A)$ is given by \cite{datta2009min, jain2002privacy} 
\begin{align}
D_{\max}(\rho||\sigma) = \min\left\{\lambda\in\mathbb{R}:2^\lambda\sigma\succeq\rho\right\}\,,
\end{align}
where $\rho\succeq\sigma$ means that $\rho-\sigma\succeq0$. Note that we have $D_{\max}(\rho||\sigma)=\widetilde{D}_{\infty}(\rho||\sigma)$. The min-relative entropy is given by
\begin{align}
D_{\min}(\rho\|\sigma) = -\log \|\sqrt{\rho} \sqrt{\sigma}\|_{1}^{2} = -\log F(\rho, \sigma)^2\,.
\end{align}
Note that we have $D_{\min}(\rho||\sigma)=\widetilde{D}_{1/2}(\rho||\sigma)$.

The quantum mutual information for $\rho_{AB}\in\cS_{\leq}(AB)$ is defined as
\begin{align}
I(A:B)_\rho = D(\rho_{AB}\|\rho_A\otimes\rho_B)
\end{align}
and the mutual information variance as
\begin{align}
V(A:B)_\rho = V(\rho_{AB}\|\rho_A\otimes\rho_B)\,.
\end{align}
The max-information that $B$ has about $A$ for $\rho_{AB}\in\cS(AB)$ is \cite{berta2011quantum}
\begin{align}
I_{\max}(A;B)_\rho = \min_{\sigma_B\in\cS(B)}D_{\max}(\rho_{AB}||\rho_A\otimes \sigma_B)\,.
\end{align}
The generalized sandwiched R\'enyi mutual information for $\rho_{AB}\in\cS(AB)$ and $\tau_A\in\cP(A)$ with $\rho_A\ll\tau_A$ is \cite{gupta2015multiplicativity}
\begin{align}
\widetilde{I}_{\alpha}\left(\rho_{A B} \| \tau_{A}\right)=\min_{\sigma_{B} \in \mathcal{S}(B)} \widetilde{D}_{\alpha}\left(\rho_{A B} \| \tau_{A} \otimes \sigma_{B}\right)\,.
\end{align}


\subsection{Smoothed entropic quantities}
The quantum hypothesis testing relative entropy is defined for $\varepsilon\in [0,1)$, $\rho\in\cS(A)$, and $\sigma\in\cP(A)$ as \cite{hiai1991proper, nagaoka2000strong, wang2012one} 
\begin{align}
D_{h}^{\varepsilon}(\rho \| \sigma)=-\log \inf\limits_{\text{Tr}(\Lambda \rho) \geq 1-\varepsilon\atop0 \leq \Lambda \leq I} \text{Tr}(\Lambda \sigma)\,.
\end{align}
The information spectrum relative entropy is defined for $\varepsilon\in(0,1)$, $\rho\in\cS(A)$, and $\sigma\in\cP(A)$ as \cite{datta2014second, tomamichel2013hierarchy} 
\begin{align}
\underline{D}_{s}^{\varepsilon}(\rho \| \sigma)=\sup \Big\{\gamma \Big| \operatorname{Tr}\big(\left(\rho-2^{\gamma} \sigma\right)\left\{\rho>2^{\gamma} \sigma\right\}_{+}\big) \geq 1-\varepsilon\Big\},
\end{align}
where $\{\cdot\}_+$ is the projector on the positive part of the argument. The information spectrum entropy is defined as $\bar{H}_{s}^{\varepsilon}(\rho)=-\underline{D}_{s}^{\varepsilon}(\rho \| I)$ \cite{datta2014second}.

One defines further smooth entropy measures by extremizing over a set of states in a small neighbourhood~\cite{renner2005security}. The smooth max-relative entropy between $\rho_A\in\cS_{\leq}(A)$ and $\sigma_A\in\cS_{\leq}(A)$ is given as 
\begin{align}
D^{\varepsilon}_{\max}(\rho_{A}\|\sigma_A) =\inf\limits_{\bar{\rho}\in\cB^{\varepsilon}(\rho)} D_{\max}(\bar{\rho}_A\|\sigma_A)
\end{align}
and the smooth min-relative entropy is given as \cite{dupuis2014generalized}
\begin{align}
D^{\varepsilon}_{\min}(\rho_{A}\|\sigma_A) =\sup\limits_{\bar{\rho}\in\cB^{\varepsilon}(\rho)} D_{\min}(\bar{\rho}_A\|\sigma_A)\,.
\end{align}
The smooth max-information that $B$ has about $A$ for $\rho_{AB}\in\cS(AB)$ is given by \cite{berta2011quantum}
\begin{align}
I_{\max}^{\varepsilon}(A; B)_{\rho}=\inf\limits_{\bar{\rho} \in \mathcal{B}^{\varepsilon}\left(\rho\right)}\ I_{\max }(A;B)_{\bar{\rho}}\,.
\end{align}
We define the partially smoothed max-information as~\cite{anshu2020partially}
\begin{align}
I^{\varepsilon}_{\max}(\dot{A};B)_{\rho} = \inf_{\bar{\rho}_{AB}\in\cB^{\varepsilon}(\rho_{AB})\atop\bar{\rho}_A = \rho_A} I_{\max}(A;B)_{\bar{\rho}}\,.
\end{align}
Note that for all the above quantities, we may replace the infimum (supremum) with the minimum (maximum) since the set of subnormalized quantum states is closed and bounded and the functions being extremized are continuous.

\section{Quantum state splitting}\label{sec:state split}

\subsection{Task}

Here, we introduce quantum state splitting first formulated in~\cite{Abeyesinghe09}, with more refined versions appearing in~\cite{anshu2017quantum, anshu2020partially, berta2011quantum}.

\begin{definition}[One-shot state splitting]\label{def:state_splitting}
Let $\rho_{AB}\in \cS(AB)$ and $\varepsilon\in[0,1]$. A one-shot quantum state splitting protocol consists of:
\begin{enumerate}
    \item Quantum registers Q, K, L and $A_1\cong B$
    \item A resource $\omega_{KL}\in\cS(KL)$ 
    \item An encoding quantum channel $\cE_{AA_1K\rightarrow AQ}$
    \item A decoding quantum channel $\cD_{QL\rightarrow B}$.
\end{enumerate}
A $\{q,\varepsilon\}$-one-shot quantum state splitting protocol of $\rho_{AB}$ is such that $q = \log|Q|$ with
\begin{align}
\text{$(\cD\circ\cE)(\rho_{AA_1R}\otimes\omega_{KL})\approx_{\varepsilon}\rho_{ABR}$ for any extension $\rho_{AA_1R}$ of $\rho_{AA_1}$.}
\end{align}
The minimal quantum communication cost of state splitting $\rho_{AB}$ is defined as\footnote{Since free entanglement can be included in the resource state $\omega_{KL}$, we equivalently have that the classical communication $2q$.}
\begin{align}
q^\star_{\varepsilon}(\rho_{AB}) = \min\Big\{q\in\mathbb{N} : \exists\text{ a} \ \{q, \varepsilon\}\text{-one-shot quantum state splitting protocol of } \rho_{AB}\Big\}\,.
\end{align}
\end{definition}

\begin{remark}\label{rmk:equivalence_of_purifications_state_splitting}
Instead of considering an arbitrary extension $\rho_{ABR}$ in Definition~\ref{def:state_splitting}, one can choose any fixed purification of $\rho_{AB}$ in the analysis of quantum state splitting. Such a choice can be made without loss of generality because any two purifications are related by an isometry on the reference quantum register and such isometries commute with the protocol. Since the purified distance is non-increasing under partial trace, a $\{q,\varepsilon\}$-one-shot quantum state splitting protocol of a fixed purification is also a $\{q,\varepsilon\}$-one-shot quantum state splitting protocol of all extensions of $\rho_{AB}$.
\end{remark}

An alternative but equivalent description of quantum state splitting is through a quantum channel $\cT^\sigma_{AA_1\rightarrow AB}$ such that
\begin{align}\label{eqn:alternate_state_splitting_description}
    \cT^{\sigma} = \cD\circ\cE\circ\cP^\sigma
\end{align}
where $\cP^\sigma$ is the preparation quantum channel of the resource state $\sigma_{KL}$. We note that $\cT^{\sigma}(\cdot) = (\cD\circ\cE)(\cdot\otimes\sigma_{KL})$. The minimal quantum communication cost is then quantified as follows.

\begin{theorem}[One-shot quantum state splitting \cite{anshu2017quantum, anshu2020partially}]\label{thm:state_splitting_smooth_imax}
Let $\rho_{AB}\in\cS(AB)$, $\varepsilon\in(0,1]$, and $\delta\in(0,\varepsilon]$. Then, the minimal quantum communication cost of $\varepsilon$-error quantum state splitting of $\rho_{AB}$ is bounded as
\begin{align}
   \frac{1}{2}I^\varepsilon_{\max}(\dot{R};B)_{\rho} \leq q^{\star}_{\varepsilon}(\rho_{AB}) \leq \frac{1}{2}I^{\varepsilon - \delta}_{\max}(\dot{R};B)_{\rho} + \log\frac{2}{\delta}\,,
\end{align}
where $\rho_{ABR}\in\cS(ABR)$ is any purification of $\rho_{AB}$.
\end{theorem}

For completeness we include the proof for our exact setting in Appendix~\ref{app:state_splitting_proofs}. From Theorem~\ref{thm:state_splitting_smooth_imax}, one directly obtains the quantum communication cost of quantum state splitting of i.i.d.\ states in the asymptotic limit. Namely, by Lemma~\ref{lem:theorem4_anshu2020partially}, one can bound the partially smoothed max-information with the smoothed max-information, and then the asymptotic equipartition property for the smoothed max-information states that \cite[Corollary B.22]{berta2011quantum}
\begin{align}
\lim _{\varepsilon \rightarrow 0} \lim _{n \rightarrow \infty} \frac{1}{n} I_{\max }^{\varepsilon}(A: B)_{\rho^{\otimes n}}=I(A: B)_{\rho}\,.
\end{align}
Hence, the minimal quantum communication cost $q^\star_\varepsilon(\rho^{\otimes n}_{AB})$ of $\varepsilon$-error one-shot quantum state splitting of $\rho_{AB}^{\otimes n}$ satisfies
\begin{align}\label{eq:splitting-iid}
\lim\limits_{\varepsilon\rightarrow 0}\lim\limits_{n\rightarrow\infty} \frac{q^\star_\varepsilon(\rho^{\otimes n}_{AB})}{n} = \frac{1}{2}I(R:B)_\rho\,.
\end{align}


\subsection{Moderate deviation expansion}

We investigate as a refinement of Eq.~\eqref{eq:splitting-iid}, namely, the minimal quantum communication cost of one-shot state splitting protocols for i.i.d.\ states in the low-error moderate deviation regime. We start with moderate sequences as defined in~\cite{chubb2017moderate} and additionally define strictly moderate sequences.

\begin{definition}[Moderate sequences]\label{def:moderate_sequence}
A sequence of non-negative numbers $\{a_n\}$ with $n\in\mathbb{N}$ is a moderate sequence if $\lim_{n\rightarrow\infty}a_n = 0$ and $\lim_{n\rightarrow\infty}na^2_n=\infty$. A strictly moderate sequence is a moderate sequence that additionally satisfies $\lim _{n \rightarrow \infty} \frac{\log n}{n a_n^2}=0$. 
\end{definition}

In particular, the sequence $\{n^{\alpha}\}$ for $n\in\mathbb{N}$ and $-1/2<\alpha<0$ is a strictly moderate sequence. For $\varepsilon_n = e^{-a_n^2n}$ for some moderate sequence $\{a_n\}$, we investigate $\varepsilon_n$-one-shot quantum state splitting protocols of $\rho_{AB}^{\otimes n}$ and investigate how $q^\star_{\varepsilon_n}(\rho^{\otimes n})$ behaves as a function of $n$.

In the moderate deviation regime, the communication cost in our derivations will contain higher order terms that depend on $a_n$. The lemma below shows that for moderate sequences, the error can be multiplied by a constant and still leave the communication cost unchanged up to $o(a_n)$ terms. Similarly, if the error is multiplied by a $\text{poly}(n)$ factor for strictly moderate sequences, the communication cost remains unchanged up to $o(a_n)$ terms.

\begin{lemma}\label{lem:moderate_seq_constant_or_polyn}
For $k>0$, moderate sequences $\{a_n\}$, and $\varepsilon_n = e^{-a_n^2n}$, it holds that
\begin{align}
k\varepsilon_n = \exp\left(-n\left(a^2_n - \frac{\log k}{n}\right)\right) = \exp\left(-nb_n^2\right)
\end{align}
for some moderate sequence $\{b_n\}$. In particular, for $\eta > 0$, there exists $n\in\mathbb{N}$ sufficiently large such that $b_n \leq a_n + \eta a_n$. If $\{a_n\}$ is a strictly moderate sequence, then $\text{poly(n)}\varepsilon_n = \exp(-nb_n^2)$ for a moderate sequence $\{b_n\}$ and for $\eta'>0$, there exists sufficiently large $n\in\mathbb{N}$ such that $b_n \leq a_n + \eta' a_n$.
\end{lemma}

\begin{proof}
For the first case we have that for sufficiently large $n$
\begin{align}
b_n = \sqrt{a_n^2 - \frac{\log k}{n}}=a_n\left(\sqrt{1 - \frac{\log k}{na_n^2}}\right)\leq a_n\left(1 - \frac{\log k}{2na_n^2}\right)\leq a_n + \eta a_n\,,
\end{align}
where the first inequality uses $\sqrt{1-x}\leq 1 - \frac{x}{2}$ for $x\leq 1$ and the second inequality uses the fact that $\lim_{n\rightarrow\infty}na_n^2 = \infty$. The second case follows the same argument except that we replace the constant $k$ with $\text{poly}(n)$ and note that for strictly moderate sequences, $\lim_{n\rightarrow\infty}\frac{\log n^r}{na_n^2} = 0$ for any constant $r\in\mathbb{R}$.
\end{proof}

\vspace{5pt}

In the following, the main idea is to obtain an expansion for the partially smoothed max-information of i.i.d.\ states. This is done by bounding the partially smoothed max-information with the hypothesis testing relative entropy and then using the expansion of the hypothesis testing relative entropy in the moderate deviation regime. Our main result in this section is the following characterization.

\begin{theorem}[State splitting moderate deviation]\label{thm:state_splitting_mod_devation}
The minimal quantum communication cost of $\varepsilon_n$-error state splitting of $\rho_{AB}^{\otimes n}$ for $\rho_{AB}\in\cS(AB)$ with $\varepsilon_n = e^{-na^2_n}$ for a moderate sequence $\{a_n\}$ has the asymptotic expansion
\begin{align}
\frac{1}{n}q^\star_{\varepsilon_n}(\rho_{AB}^{\otimes n}) =\;& \frac{1}{2}I(R:B)_\rho +  a_n\sqrt{V(R:B)_\rho}+ o(a_n)\,,
\end{align}
where $\rho_{ABR}$ is any purification of $\rho_{AB}$.
\end{theorem}

The proof is developed in the following subsections.


\subsection{Technical tools}

As the first step for the proof of Theorem \ref{thm:state_splitting_mod_devation}, we clarify when the tighter triangle inequality for the purified distance~\cite[Proposition 3.16]{tomamichel2015quantum} holds. The following lemmas may be of independent interest since this tighter triangle inequality yields tight bounds between the smoothed min-relative entropy and the smoothed max-relative entropy, as well as between the smoothed min-relative entropy and the hypothesis testing relative entropy.

\begin{lemma}[Tight triangle inequality purified distance]\label{lem:purified_dist_triangle_ineq}
For $\rho,\sigma,\tau \in \cS(A)$ we have
\begin{align}\label{eq:tighter-triangle}
P(\rho,\sigma)^2 + P(\sigma,\tau)^2 \leq 1\quad\Rightarrow\quad P(\rho,\tau) \leq P(\rho,\sigma)F(\sigma,\tau) + P(\sigma,\tau)F(\rho,\sigma)\,.
\end{align}
\end{lemma}

\begin{proof}
The proof of~\cite[Proposition 3.16]{tomamichel2015quantum} shows that Eq.~\eqref{eq:tighter-triangle} holds when $\sin^{-1}(P(\rho,\sigma)) + \sin^{-1}(P(\sigma,\tau)) \leq \pi/2$. We show that this condition is equivalent to $P(\rho,\sigma)^2 + P(\sigma,\tau)^2 \leq 1$. By the monotonicity of the cosine function on $[0,\pi]$, we have
\begin{align}
    &\sin^{-1}(P(\rho,\sigma)) + \sin^{-1}(P(\sigma, \tau)) \leq \frac{\pi}{2} \\
    \iff &\cos (\sin^{-1}(P(\rho,\sigma))+\sin^{-1}(P(\sigma, \tau))) \geq \cos \left(\frac{\pi}{2}\right) \\
    \iff &\cos (\sin^{-1}(P(\rho,\sigma)))\cos(\sin^{-1}(P(\sigma, \tau))) - \sin (\sin^{-1}(P(\rho,\sigma)))\sin(\sin^{-1}(P(\sigma, \tau))) \geq 0 \\
    \iff &\sqrt{1-P(\rho,\sigma)^{2}} \sqrt{1-P(\sigma, \tau)^{2}}-P(\rho,\sigma) P(\sigma, \tau) \geq 0 \\
    \iff &P(\rho,\sigma)^2 + P(\sigma, \tau)^2 \leq 1\,,
\end{align}
where the third line follows by using the addition formula for cosine and the fourth line uses that $\cos^2(\theta) = 1- \sin^2(\theta)$. 
\end{proof}
\vspace{5pt}

The following lemma was already known for $\rho_{AB},\sigma_{AB} = I_A\otimes\rho_B\in\cS(AB)$ \cite[Remark 5.6]{tomamichel2012framework} but we get it for general quantum states.

\begin{lemma}\label{lem:dmax_dmin_bound}
Let $\rho\in\cS(A)$, $\sigma\in\cP(A)$, and $\varepsilon,\varepsilon'\in[0,1]$ with $\varepsilon^2 + \varepsilon'^2 \leq 1$. Then, we have that
\begin{align}
    D^\varepsilon_{\min}(\rho\|\sigma) &\leq D^{\varepsilon'}_{\max}(\rho\|\sigma) - \log\left(1 - \left(\varepsilon\sqrt{1-\varepsilon'^2}+\varepsilon'\sqrt{1-\varepsilon^2}\right)^2\right)\,.
\end{align}
\end{lemma}

\begin{proof}
Let $\lambda = D^{\varepsilon'}_{\max}(\rho\|\sigma) = D_{\max}(\tilde{\rho}\|\sigma)$ for $\tilde{\rho}\approx_{\varepsilon'} \rho$, it then holds that $\tilde{\rho}\leq 2^\lambda\sigma$. For some $\bar{\rho}\approx_{\varepsilon}\rho$ we have that
\begin{align}
    D^{\varepsilon}_{\min}(\rho\|\sigma) = D_{\min}(\bar{\rho}\|\sigma) &= -\log F(\bar{\rho},\sigma)^2 \\
    &\leq -\log F(\bar{\rho}, 2^{-\lambda}\tilde{\rho})^2 \\
    &= \lambda -\log F(\bar{\rho},\tilde{\rho})^2 \\
    &= D^{\varepsilon'}_{\max}(\rho\|\sigma) - \log(1 - P(\bar{\rho},\tilde{\rho})^2) \\
    &\leq D^{\varepsilon'}_{\max}(\rho\|\sigma) - \log\left(1 - \left(\varepsilon\sqrt{1-\varepsilon'^2}+\varepsilon'\sqrt{1-\varepsilon^2}\right)^2\right)\,,
\end{align}
where the last step is via the tighter triangle inequality for the purified distance (Lemma \ref{lem:purified_dist_triangle_ineq}) and a basic monotonicity argument (Lemma \ref{lem:purified_dist_bound_monotone}).
\end{proof}
\vspace{5pt}

Note that the constraint on the choices of $(\varepsilon, \varepsilon')$ is less stringent compared to \cite[Remark 5.6]{tomamichel2012framework}, which is due to the tighter triangle inequality. Next, we have the following upper bond on the min-relative entropy in terms of the hypothesis testing relative entropy.

\begin{lemma}\label{lem:dh_dmin_bound}
Let $\rho,\sigma\in\cP(A)$ and $\varepsilon\in\left(0,k^{-1/2}\right]$ for $k>1$. Then, we have that
\begin{align}
D^\varepsilon_{\min}(\rho\|\sigma) \leq D^{k\varepsilon^2}_{h}(\rho\|\sigma) - \log\left(1 - \left(\varepsilon^2\sqrt{k}+\sqrt{1- k\varepsilon^2}\sqrt{1-\varepsilon^2}\right)^2\right)\,.
\end{align}
\end{lemma}

\begin{proof}
We choose $\varepsilon' = \sqrt{1 - k\varepsilon^2}$ in the preceeding Lemma \ref{lem:dmax_dmin_bound} to obtain
\begin{align}
    D^\varepsilon_{\min}(\rho\|\sigma) \leq D^{\sqrt{1 - k\varepsilon^2}}_{\max}(\rho\|\sigma) - \log f(\varepsilon) \quad\text{for}\quad f(\varepsilon) = \log\left(1 - \left(\varepsilon^2\sqrt{k}+\sqrt{1- k\varepsilon^2}\sqrt{1-\varepsilon^2}\right)^2\right)\,.
\end{align}
By a connection of the smooth max-relative entropy to the hypothesis testing relative entropy (Lemma~\ref{lem:thm4_anshu2019minimax}), we have for $\delta\in(0,1)$ 
\begin{align}
    D^{\sqrt{\delta}}_{\max}(\rho\|\sigma) \leq D_h^{1-\delta}(\rho\|\sigma)  - \log(1-\delta)\,.
\end{align}
Combining the two inequalities, we get
\begin{align}
    D^\varepsilon_{\min}(\rho\|\sigma) \leq D^{\sqrt{1 - k\varepsilon^2}}_{\max}(\rho\|\sigma) - \log f(\varepsilon)&\leq D_h^{1-(\sqrt{1-k\varepsilon^2})^2}(\rho\|\sigma) - \log f(\varepsilon) - \log k\varepsilon^2  \\
    &\leq D_h^{k\varepsilon^2}(\rho\|\sigma) - \log f(\varepsilon) - \log k\varepsilon^2\,.
\end{align}
\end{proof}
\vspace{5pt}

Finally, we recall the moderate deviation expansion for the hypothesis testing relative entropy of i.i.d.\ states.

\begin{prop}[Hypothesis testing AEP~\cite{chubb2017moderate, cheng2017moderate, watanabe2017finite}]\label{prop:chubbs_dh_expansion}
For any moderate sequence $\{a_n\}$, let $\varepsilon_n = e^{-a_n^2n}$. For $\rho,\sigma\in\cS(A)$ with $\rho\ll \sigma$ and $\eta_1,\eta_2 > 0$, there exists $n^\star\in\mathbb{N}$ such that for $n\geq n^\star$ we have
\begin{align}
    \frac{1}{n} D_{\mathrm{h}}^{\varepsilon_{n}}\left(\rho^{\otimes n} \| \sigma^{\otimes n}\right)&\leq D(\rho \| \sigma)-\sqrt{2 V(\rho \| \sigma)} a_{n} + \eta_1 a_{n}. \\
    \frac{1}{n} D_{\mathrm{h}}^{1-\varepsilon_{n}}\left(\rho^{\otimes n} \| \sigma^{\otimes n}\right)&\leq D(\rho \| \sigma)+\sqrt{2 V(\rho \| \sigma)} a_{n}+\eta_2 a_{n}\,.
\end{align}
\end{prop}


\subsection{Proof of Theorem \ref{thm:state_splitting_mod_devation}}

We now have the main technical results in hand to derive our moderate deviation analysis: An expansion for the partially smoothed max-information of i.i.d.\ states in the low-error moderate deviation regime.

\begin{prop}[Partially smoothed max-information AEP]\label{prop:smooth_imax_expansion}
Let $\varepsilon_n = e^{-na^2_n}$ for a moderate sequence $\{a_n\}$ and let $\rho_{BR}\in\cS(BR)$. Then, for $\eta>0$, there exists $n^\star\in\mathbb{N}$ such that for $n\geq n^\star$ we have
\begin{align}
I(B:R)_\rho + a_n\sqrt{4V(B:R)_\rho} - \eta a_n\leq \frac{1}{n}I^{\varepsilon_n}_{\max}(\dot{R}^n;B^n)_{\rho^{\otimes n}} &\leq I(B:R)_\rho + a_n\sqrt{4V(B:R)_\rho} + \eta a_n\,.
\end{align}
\end{prop}

\begin{proof}
We start by showing the upper bound. For $\eta>0$, there exists $n_1^\star$ such that for $n\geq n_1^\star$ we have
\begin{align}
    I^{\varepsilon_n}_{\max}(\dot{R}^n;B^n)_{\rho^{\otimes n}}
    &\leq I^{\frac{\varepsilon_n}{4}}_{\max}(R^n;B^n)_{\rho^{\otimes n}} +  \log\frac{8+(\frac{\varepsilon_n}{2})^2}{(\frac{\varepsilon_n}{2})^2} \\
    &\leq D^{\frac{\varepsilon_n}{4}}_{\max}(\rho_{BR}^{\otimes n}\|\rho_{B}^{\otimes n}\otimes\rho_R^{\otimes n})+  \log\frac{8+(\frac{\varepsilon_n}{2})^2}{(\frac{\varepsilon_n}{2})^2} \\
    & \leq D_h^{1 - (\frac{\varepsilon_n}{4})^2}(\rho_{BR}^{\otimes n}\|\rho_{B}^{\otimes n}\otimes\rho_R^{\otimes n})+  \log\frac{8+(\frac{\varepsilon_n}{2})^2}{(\frac{\varepsilon_n}{2})^2} \\
    &\leq nD(\rho_{BR}\|\rho_{B}\otimes\rho_R) + n\cdot a_n\sqrt{4V(\rho_{BR}\|\rho_{B}\otimes\rho_R)} + n\cdot\eta a_n\,,
\end{align}
where the first inequality follows from the bound on the partially smoothed max-information with the smoothed max-information (Lemma~\ref{lem:theorem4_anshu2020partially}), the second inequality is by choosing $\rho_R^{\otimes n}$ instead of a minimization, the third inequality follows from the bounds between the max relative entropy and the hypothesis testing relative entropy (Lemma~\ref{lem:thm4_anshu2019minimax}) and the fourth inequality follows by the AEP for the hypothesis testing relative entropy (Proposition \ref{prop:chubbs_dh_expansion}), a property of moderate deviation sequences (Lemma \ref{lem:moderate_seq_constant_or_polyn}), the boundedness of the relative entropy variance (Lemma~\ref{lem:bounded_mutual_info_variance}), and noting that the term in the logarithm is $n\cdot o(a_n)$. 

To prove the lower bound, let $\rho_{BRR'}^{\otimes n}$ be a purification of $\rho_{BR}^{\otimes n}$ with $R'\cong BR$. For any $\tilde{\rho}_{B^nR^n}\approx_{\varepsilon_n}\rho_{BR}^{\otimes n}$, there exists a purification $\tilde{\rho}_{B^nR^nR'^n}\approx_{\varepsilon_n}\rho_{BRR'}^{\otimes n}$ by Uhlmann's theorem. We have
\begin{align}
    I^{\varepsilon_n}_{\max}(\dot{R}^n;B^n)_{\rho^{\otimes n}}
    &= \inf\limits_{\tilde{\rho}_{B^nR^n}\in\cB^{\varepsilon_n}\left(\rho_{BR}^{\otimes n}\right)\atop \ \tilde{\rho}_{R^n }= \rho_R^{\otimes n}}\inf\limits_{\sigma_{B^n}} D_{\max}(\tilde{\rho}_{B^nR^n}\|\rho_R^{\otimes n}\otimes \sigma_{B^n})\\
    &\geq \inf\limits_{\tilde{\rho}\in\cB^{\varepsilon_n}\left(\rho^{\otimes n}\right)}\inf\limits_{\sigma_{B^n}} D_{\max}(\tilde{\rho}_{B^nR^n}\|\rho_R^{\otimes n}\otimes \sigma_{B^n})\\
    & = \inf_{\tilde{\rho}\in\cB^{\varepsilon_n}\left(\rho^{\otimes n}\right)} \widetilde{I}_{\infty}(\tilde{\rho}_{B^nR^n}\|\rho_R^{\otimes n}) \\
    &= \inf\limits_{\tilde{\rho}\in\cB^{\varepsilon_n}\left(\rho^{\otimes n}\right)} -\widetilde{I}_{\frac{1}{2}}(\tilde{\rho}_{R^nR'^n}\|(\rho_R^{\otimes n})^{-1}) \\
    &\geq \inf\limits_{\tilde{\rho}\in\cB^{\varepsilon_n}\left(\rho^{\otimes n}\right)} -\widetilde{D}_{\frac{1}{2}}(\tilde{\rho}_{R^nR'^n}\|(\rho_R^{\otimes n})^{-1}\otimes\rho_{R'}^{\otimes n}) \\
    &\geq - D^{\varepsilon_n}_{\min}(\rho_{RR'}^{\otimes n}\| (\rho_R^{\otimes n})^{-1}\otimes\rho_{R'}^{\otimes n})\,,
\end{align}
where the first inequality follows from relaxing the constraint on $\tilde{\rho}_{BR}$, the second equality is from the definition of the sandwiched R\'enyi mutual information, the third equality follows from the duality of the sandwiched R\'enyi mutual information (Lemma~\ref{lem:duality_renyi_relative_entropies}) and choosing a specific purification satisfying $\tilde{\rho}_{B^nR^nR'^n}\approx_{\varepsilon_n}\rho_{BRR'}^{\otimes n}$. The second inequality follows by choosing a $\rho_{R'}^{\otimes n}$ instead of a minimization over $\sigma_{R'}^n\in\cS\left(R'^{\otimes n}\right)$, and the last inequality is by the definition of the smoothed min-relative entropy. We now choose $\eta>0$ and $k>1$ to obtain some $n_2^\star$ such that for all $n\geq n_2^\star$, we have
\begin{align}
    - D^{\varepsilon_n}_{\min}(\rho_{RR'}^{\otimes n}\| (\rho_R^{\otimes n})^{-1}\otimes\rho_{R'}^{\otimes n})
    &\geq -D^{k\varepsilon^2_n}_{h}(\rho_{RR'}^{\otimes n}\| (\rho_R^{\otimes n})^{-1}\otimes\rho_{R'}^{\otimes n}) - n\cdot\frac{\eta a_n}{2}\\
    &\geq -nD(\rho_{RR'}\|\rho_R^{-1}\otimes\rho_{R'}) + n\cdot a_n\sqrt{4V(\rho_{RR'}\|\rho_R^{-1}\otimes\rho_{R'})} -n\cdot\eta a_n \\
    &= nD(\rho_{BR}\|\rho_R\otimes\rho_{B}) + n\cdot a_n\sqrt{4V(\rho_{BR}\|\rho_R\otimes\rho_{B})}- n\cdot\eta a_n\,,
\end{align}
where the first inequality follows due to Lemma \ref{lem:dh_dmin_bound} since for sufficiently large $n$, we will have $\varepsilon_n\in \left(0, k^{-1/2}\right]$ for any choice of $k>1$ as well as the fact that $\log f(\varepsilon_n) = n\cdot o(a_n)$, the second inequality is due to the AEP for the hypothesis testing relative entropy (Proposition \ref{prop:chubbs_dh_expansion}), a property of moderate deviation sequences (Lemma \ref{lem:moderate_seq_constant_or_polyn}), and the boundedness of the relative entropy variance (Lemma~\ref{lem:bounded_mutual_info_variance}). The final equality uses the duality of the relative entropy and relative entropy variance (Lemma~\ref{lem:duality_renyi_relative_entropies}). By choosing $n^\star = \max(n_1^\star, n_2^\star)$, the proposition follows.
\end{proof}
\vspace{5pt}

The AEP for the partially smoothed max-information (Proposition \ref{prop:smooth_imax_expansion}) is strengthened straightforwardly in the following corollary by noting that constant multiplicative factors on the error $\varepsilon_n$ do not affect the moderate deviation analysis and multiplicative factors that are polynomial in $n$ do not affect the strict moderate deviation analysis. Both results hold due to a property of moderate deviation sequences (Lemma \ref{lem:moderate_seq_constant_or_polyn}) and the boundedness of the relative entropy variance (Lemma~\ref{lem:bounded_mutual_info_variance}).

\begin{corollary}\label{cor:keps_in_Imax_expansion}
Proposition \ref{prop:smooth_imax_expansion} holds for $0 < \varepsilon'_n \leq 1$ with $\varepsilon'_n = \Theta(e^{-a_n^2n})$ and a moderate sequence $\{a_n\}$. That is, for $\eta>0$ there exists $n^\star\in\mathbb{N}$ such that for $n\geq n^\star$ we have
\begin{align}
I(B:R)_\rho + a_n\sqrt{4V(B:R)_\rho} - \eta a_n\leq \frac{1}{n}I^{\varepsilon_n'}_{\max}(\dot{R}^n;B^n)_{\rho^{\otimes n}} &\leq I(B:R)_\rho + a_n\sqrt{4V(B:R)_\rho} + \eta a_n\,.
\end{align}
Moreover, for a strictly moderate sequence $\{a_n\}$ and $\varepsilon'_n = \mathrm{poly}(n)e^{-a_n^2n}$, for $\eta>0$, there exists $n^\star\in\mathbb{N}$ such that for $n\geq n^\star$ we have
\begin{align}
I(B:R)_\rho + a_n\sqrt{4V(B:R)_\rho} - \eta a_n\leq \frac{1}{n}I^{\varepsilon_n'}_{\max}(\dot{R}^n;B^n)_{\rho^{\otimes n}} &\leq I(B:R)_\rho + a_n\sqrt{4V(B:R)_\rho} + \eta a_n\,.
\end{align}
\end{corollary}

\begin{proof}[Proof of Theorem \ref{lem:moderate_seq_constant_or_polyn}]
The moderate deviation analysis of quantum state splitting (Theorem \ref{thm:state_splitting_mod_devation}) is now straightforward. For any purification $\rho_{ABR}$ of $\rho_{AB}$, we obtain the cost of one-shot quantum state splitting (Theorem \ref{thm:state_splitting_smooth_imax}). By combining the AEP for the partially smoothed max-information (Proposition \ref{prop:smooth_imax_expansion}) and noting that constant multiplicative factors do not change the communication cost (Corollary~\ref{cor:keps_in_Imax_expansion}), the claim follows.
\end{proof}


\subsection{Equivalence: Quantum state merging}

Quantum state merging can be understood as a time reversed version of a quantum state splitting \cite{Abeyesinghe09}. Alice and Bob start with a state $\rho_{ABR}$, where Alice holds the $A$ quantum register, Bob holds the $B$ quantum register and $R$ is inaccessible. The goal is to use pre-shared entanglement and one-way communication to send the $B$ quantum register to Alice. An equivalent statement to the cost of one-shot quantum state splitting (Theorem \ref{thm:state_splitting_smooth_imax}) was shown for one-shot quantum state merging in~\cite[Theorem 6]{anshu2020partially}. Hence, our moderate deviation analysis for quantum state splitting (Theorem \ref{thm:state_splitting_mod_devation}) also holds for quantum state merging. 


\section{Quantum source coding}\label{sec:extensions}

Here, we investigate the task of quantum source coding in the moderate deviation regime. There are two main settings for source coding, namely the ensemble and purification settings based on whether one uses the fidelity or entanglement fidelity to measure the success of source coding. The ensemble setting is further divided into the blind and visible settings, where the state to be compressed is either known or unknown. More precise definitions of these settings can be found in the literature~\cite{schumacher1995quantum, hayashi2009information, hayashi2002exponents, hayashi2002simple}. Here, we consider quantum source coding in the blind purification scheme, which is equivalent to quantum state splitting with a trivial $A$ quantum register and no entanglement-assistance.

\begin{definition}[Quantum source coding]\label{def:source_coding}
Let $\rho_{B}\in \cS(B)$ and $\varepsilon\in[0,1]$. A one-shot quantum source coding protocol consists of:
\begin{enumerate}
    \item Quantum registers Q and $A_1\cong B$
    \item An encoding quantum channel $\cE_{A_1\rightarrow Q}$
    \item A decoding quantum channel $\cD_{Q\rightarrow B}$.
\end{enumerate}
A $\{q,\varepsilon\}$-one-shot quantum source coding protocol of $\rho_{B}$ is such that $q = \log|Q|$ with
\begin{align}
\text{$(\cD\circ\cE)(\rho_{A_1R})\approx_{\varepsilon}\rho_{BR}$ for any extension $\rho_{A_1R}$ of $\rho_{A_1}$.}
\end{align}
The minimal quantum communication cost of source coding of $\rho_B$ is defined as
\begin{align}
q^\star_{\varepsilon}(\rho_{B}) = \min\Big\{q\in\mathbb{N} : \exists\text{ a} \ \{q, \varepsilon\}\text{-one-shot quantum source coding protocol of } \rho_{B}\Big\}\,.
\end{align}
\end{definition}

We start with the following lemma about the mutual information variance for pure states.

\begin{lemma}[Pure state mutual information variance]\label{lem:mutual_info_variance}
For $\rho_{AB}\in\cS(AB)$ pure, the mutual information variance and the conditional entropy variance $V(\rho_{AB}\|I_A\otimes\rho_B)$ take the form
\begin{align}
V(A:B)_\rho = 4V(A)_\rho,\qquad V(\rho_{AB}\|I_A\otimes\rho_B) &= V(A)_\rho\,,
\end{align}
respectively.
\end{lemma}

\begin{proof}
From the definition of the relative entropy variance, we have
\begin{align}
&V(\rho_{AB}\|\rho_A\otimes\rho_B)\nonumber\\
&= \Tr{\rho_{AB}(\log\rho_{AB})^2} + \Tr{\rho_{A}(\log\rho_{A})^2} + \Tr{\rho_{B}(\log\rho_{B})^2} + 2\Tr{\rho_{AB}\log\rho_A\otimes\log\rho_B} \nonumber \\ 
&\quad- 2\Tr{\rho_{AB}\log\rho_{AB}\log\rho_A} -2\Tr{\rho_{AB}\log\rho_{AB}\log\rho_B} - D(\rho_{AB}\|\rho_A\otimes\rho_B)^2 \\
&= \Tr{\rho_{A}(\log\rho_{A})^2} + \Tr{\rho_{B}(\log\rho_{B})^2} + 2\Tr{\rho_{AB}\log\rho_A\otimes\log\rho_B} - (S(A)_\rho + S(B)_\rho)^2 \\
&= 4V(A)_\rho\,,
\end{align}
where in the second equality we have removed various terms which are zero due to $\rho_{AB}$ being pure, and in the third equality we have used that the marginals $\rho_A$ and $\rho_B$ share the same eigenvalues and that $\Tr{\rho_{AB}\log\rho_A\otimes\log\rho_B} = \Tr{\rho_{A}(\log\rho_{A})^2}$ for pure $\rho_{AB}$. The last claim can be verified by expanding the left hand side in the Schmidt basis of $\rho_{AB}$. A similar argument shows the second equality in the lemma. 
\end{proof}
\vspace{5pt}

We now consider quantum source coding in the moderate deviation regime. Note that quantum source coding in the large deviation setting was explored in~\cite{hayashi2002exponents, hayashi2009information} using the information spectrum method. Remark $16$ of~\cite{nagaoka2007information} also touches on this connection between source coding and information spectrum. Classical source coding in the moderate deviation regime was investigated in~\cite{hayashi2020finite}. Some of our results below may also be proved using a different route, namely the moderate deviation expansion of the information spectrum entropy~\cite[Section 3.7]{dembo2009large}. The minimal communication cost of quantum source coding in the low-error moderate deviation regime is characterized as follows.

\begin{theorem}[Source coding moderate deviation]\label{thm:source_coding_ea_low_error}
The minimal quantum communication cost of one-shot $\varepsilon_n$-error source coding of $\rho_{B}^{\otimes n}$ for $\rho_{B}\in\cS(B)$ with $\varepsilon_n = e^{-na^2_n}$ for a moderate sequence $\{a_n\}$ is given by
\begin{align}
    \frac{1}{n}q^\star_{\varepsilon_n}(\rho_{B}^{\otimes n}) =& S(B)_\rho +  2a_n\sqrt{V(B)_\rho} + o(a_n)\,.
\end{align}
\end{theorem}

\begin{proof}
The converse statement follows as a corollary of our moderate deviation analysis of state splitting (Theorem \ref{thm:state_splitting_mod_devation}) and properties of the mutual information variance for pure states (Lemma \ref{lem:mutual_info_variance}). While a related achievability statement also follows from those results, this additionally assumes the presence of entanglement-assistance between the encoder and decoder which is not present in Definition~\ref{def:source_coding}. Thus, we use a different proof technique that does not require entanglement assistance. Namely, we employ the bounds shown in~\cite{datta2014second, abdelhadi2020second}, where the one-shot quantum source coding error is quantified in terms of the fidelity. We modify their result for our setting in purified distance. For $\varepsilon_n = e^{-a_n^2n}$ for a moderate sequence $\{a_n\}$, we have
\begin{align}
\frac{1}{n}q^\star_{\varepsilon_n}(\rho_{B}^{\otimes n}) \leq \frac{1}{n}\bar{H}_{s}^{\frac{\varepsilon^2_n}{2}}(\rho^{\otimes n}_{B}) &= -\frac{1}{n}\underline{D}_{s}^{\frac{\varepsilon^2_n}{2}}(\rho^{\otimes n}_{B}\| I_{B^n}) \\
&\leq -\frac{1}{n}D_h^{\varepsilon^2_n/4}(\rho^{\otimes n}_{B}\| I_{B^n}) - \frac{1}{n}\log\frac{\varepsilon^2_n}{4} \\
&= -D(\rho_{B}\|I_{B}) + a_n\sqrt{4V(\rho_{B}\| I_{B})} + o(a_n) \\
&= D(\rho_{BR}\|I_{B}\otimes \rho_R) + a_n\sqrt{4V(\rho_{BR}\| I_{B}\otimes\rho_R)} + o(a_n) \\
&= S(B)_\rho + 2a_n\sqrt{V(B)_\rho} + o(a_n)\,,
\end{align}
where the first inequality is due to the achievability bound in~\cite[Theorem 5.5 (ii)]{datta2014second}, the first equality is from the definition of the information spectrum entropy, the second inequality is from the bound between the hypothesis testing relative entropy and the information spectrum relative entropy (Lemma~\ref{lem:prop4.7_datta2014second}), the second equality is due to the AEP for the hypothesis testing relative entropy (Proposition \ref{prop:chubbs_dh_expansion}) and because multiplicative factors on the error do not change the communication cost in the moderate deviation regime (Corollary \ref{cor:keps_in_Imax_expansion}), and the third equality is through the duality of the relative entropy (Lemma~\ref{lem:duality_renyi_relative_entropies}) and the final equality is due to properties of the mutual information variance for pure states (Lemma \ref{lem:mutual_info_variance}).
\end{proof}
\vspace{5pt}

We note that the converse proof allows for entanglement-assistance. Hence, it follows that entanglement-assistance does not reduce the quantum communication cost of low-error moderate deviation quantum source coding.

In the large error regime, we set $\varepsilon_n = 1- e^{-na_n^2}$ for a moderate sequence $\{a_n\}$. We find the following achievability result.

\begin{prop}[Source coding high-error moderate deviation]\label{prop:source_coding_ea_high_error}
The minimal quantum communication cost of one-shot $\varepsilon_n$-error source coding of $\rho_{B}^{\otimes n}$ for $\rho_{B}\in\cS(B)$ with $\varepsilon_n = 1- e^{-na^2_n}$ for a moderate sequence $\{a_n\}$ is bounded as
\begin{align}
    \frac{1}{n}q^\star_{\varepsilon_n}(\rho_{B}^{\otimes n}) &\leq S(B)_\rho -  a_n\sqrt{V(B)_\rho} + o(a_n)\,.
\end{align}
\end{prop}

\begin{proof}
We start with the achievability result. We have
\begin{align}
\frac{1}{n}q^\star_{\varepsilon_n}(\rho_{B}^{\otimes n}) \leq \frac{1}{n}\bar{H}_{s}^{1-\sqrt{1 - \varepsilon_n^2}}(\rho^{\otimes n}_{B}) &= \frac{1}{n}\bar{H}_{s}^{1-\sqrt{2e^{-na_n^2} - e^{-2na_n^2}}}(\rho^{\otimes n}_{B})\\ 
&= -\frac{1}{n}\underline{D}_{s}^{1-\sqrt{2e^{-na_n^2} - e^{-2na_n^2}}}(\rho^{\otimes n}_{B}\| I_{B^n}) \\
&\leq -\frac{1}{n}D_h^{1-\sqrt{2e^{-na_n^2} - e^{-2na_n^2}}-\delta}(\rho^{\otimes n}_{B}\| I_{B^n}) - \frac{1}{n}\log\delta \\
&\leq -\frac{1}{n}D_h^{1-\sqrt{2e^{-na_n^2}}-\delta}(\rho^{\otimes n}_{B}\| I_{B^n}) - \frac{1}{n}\log\delta\\
&= -D(\rho_{B}\|I_{B}) - a_n\sqrt{V(\rho_{B}\| I_{B})} + o(a_n) \\
&= D(\rho_{BR}\|I_{B}\otimes \rho_R) - a_n\sqrt{V(\rho_{BR}\| I_{B}\otimes\rho_R)} + o(a_n) \\
&= S(B)_\rho - a_n\sqrt{V(B)_\rho} + o(a_n)\,,
\end{align}
where the first inequality is due to the achievability result \cite[Eq.(31)]{abdelhadi2020second}, the first equality substitutes $\varepsilon_n = 1-e^{-na_n^2}$, the second equality is from the definition of the information spectrum entropy, the second inequality is from the bound between the hypothesis testing relative entropy and the information spectrum relative entropy (Lemma~\ref{lem:prop4.7_datta2014second}), and the third inequality is due to the monotonicity of the hypothesis testing relative entropy in the $\varepsilon$-parameter. The third equality is due to the AEP for the hypothesis testing relative entropy (Proposition \ref{prop:chubbs_dh_expansion}), because multiplicative factors on the error do not change the communication cost in the moderate deviation regime (Corollary \ref{cor:keps_in_Imax_expansion}), and the choice of $\delta = \sqrt{e^{-na_n^2}}$. The fourth equality is through the duality of the relative entropy (Lemma~\ref{lem:duality_renyi_relative_entropies}), and the final equality is due to the properties of the mutual information variance for pure states (Lemma \ref{lem:mutual_info_variance}). 
\end{proof}
\vspace{5mm}

Based on the following discussion, we conjecture that the achievability statement in Proposition~\ref{prop:source_coding_ea_high_error} is also tight. For $\varepsilon\in(0,1)$ and $\rho_{B}\in\cS(B)$, it was shown in the second-order analysis of quantum source coding that \cite{abdelhadi2020second}
\begin{align}
    \frac{1}{n}q^\star_{\varepsilon}(\rho^{\otimes n}) = S(B)_\rho +\sqrt{\frac{V(B)_\rho}{n}}\phi^{-1}(\sqrt{1-\varepsilon^2}) + O\left(\frac{\log n}{n}\right)\, ,
\end{align}
where $\phi(\cdot)$ is the cumulative distribution function of the standard normal distribution. We may substitute $\varepsilon = e^{-na_n^2}$ or $\varepsilon = 1 - e^{-na_n^2}$ and use the approximation $\phi^{-1}(\delta) \approx -\sqrt{-2\log\delta}$ for small $\delta$ and that $\phi^{-1}(\delta) = -\phi^{-1}(1-\delta)$ for $\delta\in[0,1]$. With this approximation and ignoring higher order error terms, we recover the low-error regime result in Theorem~\ref{thm:source_coding_ea_low_error}, as well as the high-error regime achievability result in Proposition~\ref{prop:source_coding_ea_high_error}. 


\section{Quantum channel simulation}\label{sec:channel_sim_mod_dev}

Quantum state splitting is the underlying primitive to understand the task of quantum channel simulation. In this setting, Alice and Bob have access to pre-shared resource states and noiseless one-way communication which they will use to simulate i.i.d.\ copies of a noisy quantum channel. This problem is also known as the Quantum Reverse Shannon Theorem in the literature~\cite{bennett2009quantum, berta2011quantum}. We start with the measure we need to distinguish between two quantum channels. It is based on the channel fidelity $F(\cE, \cF) = \max_{\rho_{AR}} F((\cE\otimes \cI_R)(\rho_{AR}), (\cF\otimes \cI_R)(\rho_{AR}))$, but we choose to work with the purified distance instead of the fidelity.

\begin{definition}[Channel purified distance]\label{def:purified_diamond_distance_original}
For any pair of quantum channels $\cE_{A\rightarrow B}$ and $\cF_{A\rightarrow B}$, the channel purified distance is defined as
\begin{align}
P(\cE,\cF) = \sup_{\rho\in\cS_{\leq}(A\otimes R)} P\big((\cE_{A\to B}\otimes \mathcal{I}_R)(\rho_{AR}) , (\cF_{A\to B}\otimes \mathcal{I}_{R})(\rho_{AR})\big)
\end{align}
where the supremum is over quantum registers $R$.
\end{definition}

We now comment on a few properties that allow us to simplify Definition \ref{def:purified_diamond_distance_original}. We can choose $|R| = |A|$ without loss of generality, as shown in Lemma~\ref{lem:bound_dim_R_purified_diamond_dist}. The supremum over $\rho$ can then be replaced by a maximum since the set of subnormalized quantum states is closed and bounded for fixed $R$ and the purified distance is a continuous function (since continuity of the one-norm implies continuity of the fidelity). By Lemma~\ref{lem:purified_dist_scaling}, we have that normalized states will achieve the maximum. Next, since the purified distance is non-increasing under partial trace, we may choose $\rho_{AR}$ to be pure. Finally, we remark that the channel purified distance is a metric since it inherits this property from the purified distance.

One-shot quantum channel simulation is as follows.

\begin{definition}[One-shot quantum channel simulation]\label{def:one_shot_channel_sim}
Let $\cN_{A\rightarrow B}$ be a quantum channel, $\varepsilon\in[0,1]$, and take a quantum register $Q$ with $\log|Q| = q$. A $\{q,\varepsilon\}$-one-shot quantum channel simulation of $\cN_{A\rightarrow B}$ consists of
\begin{enumerate}
    \item A resource $\sigma_{KL}\in\cS(KL)$ 
    \item An encoding quantum channel $\cE_{AK\rightarrow Q}$
    \item A decoding quantum channel $\cD_{QL\rightarrow B}$
\end{enumerate}
such that we have for the composite quantum channel
\begin{align}
\cN'_{A\rightarrow B} = \cD\circ\cE\circ\cP^{\sigma} = (\cD\circ\cE)(\cdot\otimes\sigma_{KL})\quad\text{that}\quad P(\cN,\cN')\leq \varepsilon\,.
\end{align}
The minimal quantum communication cost of simulating $\cN_{A\rightarrow B}$ is defined as\footnote{Since free entanglement can be included in the resource state $\sigma_{KL}$, we equivalently have that the classical communication required is $2q$.}
\begin{align}
q^\star_{\varepsilon}(\cN) = \min\Big\{q: \exists \ \text{a } \{q, \varepsilon\}\text{-one-shot quantum channel simulation of } \cN\Big\}\,.
\end{align}
\end{definition}

When the channel being simulated is of the form $\cN^{\otimes n}$, we can make use of de Finetti reductions and symmetrization to achieve a one-shot quantum channel simulation protocol in terms of a one-shot quantum state splitting protocol. We can then employ our quantum state splitting results to characterize the minimal communication cost of quantum channel simulation in the low-error strictly moderate deviation setting. Our main result in this section is as follows.

\begin{theorem}[Channel simulation moderate deviation]\label{thm:channel_sim_mod_dev_final_result}
For any quantum channel $\cN_{A\rightarrow B}$, the minimal quantum communication cost for $\cN^{\otimes n}$ with error $\varepsilon_n = e^{-na_n^2}$ for a strictly moderate sequence $\{a_n\}$ is
\begin{align}
\frac{1}{n}q^\star_{\varepsilon_n}(\cN^{\otimes n}) = C(\cN) + a_n\sqrt{V_{\max}(B:R)_{\cN}} + o(a_n)\,,
\end{align}
where $V_{\max}(B:R)_{\cN}$ was defined in Eq.~\eqref{eq:channel-capacities}.
\end{theorem}

For the rest of this section, we go through the technical tools for the proof of Theorem \ref{thm:channel_sim_mod_dev_final_result}. The following lemma about symmetrizing quantum protocols is standard in the literature (see, e.g., \cite{berta2011quantum, bennett2009quantum}). However, we have not found a formal version, which is why we state it here and give a proof in Appendix \ref{app:technical}.

\begin{lemma}[Symmetrized protocol for state splitting]\label{lem:perm_inv_state_splitting}
Let $\pi$ be a permutation of a set of $n$ elements. Let $\rho_{A^nB^n}\in\cS\left(A^nB^n\right)$ such that for any permutation quantum channel $\pi_{A^n}\otimes\pi_{B^n}: A^nB^n\rightarrow A^nB^n$ where $\pi_{A^n}(\pi_{B^n})$ applies the permutation $\pi$ to the $A^n$($B^n$) registers, we have
\begin{align}
\pi_{A^n}\otimes \pi_{B^n}(\rho_{A^nB^n}) = \rho_{A^nB^n}\,.
\end{align}
Furthermore, let there be a $\{q,\varepsilon\}$-one-shot quantum state splitting protocol of $\rho_{A^nB^n}$ given in terms of the quantum channel $\cT^{\sigma}_{A^nA_1^n\rightarrow A^nB^n}$ and some resource state $\sigma$. Then, there exists a one-shot quantum state splitting protocol given in terms of a quantum channel $\bar{\cT}^{\omega}_{A^nA_1^n\rightarrow A^nB^n}$ such that 
\begin{align}
    \bar{\cT}^{\omega}\circ\bar{\pi}_{A^n}\otimes\bar{\pi}_{A_1^n} = \bar{\pi}_{A^n}\otimes\bar{\pi}_{B^n}\circ \bar{\cT}^{\omega} 
\end{align}
for any permutation $\bar{\pi}$ on $n$ quantum registers, $\omega = \sigma\otimes\sigma'$, and $\sigma'$ is a state containing shared randomness of dimension $n!$. The quantum channel $\bar{\cT}^{\omega}_{A^nA_1^n\rightarrow A^nB^n}$ satisfies
\begin{align}
P(\rho_{A^nB^nR}, (\bar{\cT}^{\omega}\otimes \cI_R)(\rho_{A^nA_1^nR})) \leq P(\rho_{A^nB^nR}, (\cT^\sigma\otimes \cI_R)(\rho_{A^nA_1^nR}))
\end{align}
for any extension $\rho_{A^nB^nR}$ of $\rho_{A^nB^n}$.
\end{lemma}

We now state the achievability and converse statements for one-shot channel simulation of $n$-fold i.i.d.\ channels. Both are based on well-known arguments, but adapted to our exact setting.

\begin{prop}[Channel simulation achievability]\label{prop:achievability_channel_sim}
Let $\cN_{A\rightarrow B}$ be a quantum channel and $\varepsilon \in (0, 1]$, and $n\in\mathbb{N}$. Then, we have
\begin{align}
&q^\star_{\varepsilon}(\cN^{\otimes n}) \leq  \max\limits_{\phi_{AR}} \frac{1}{2}I^{\frac{\varepsilon'}{72}}_{\max}(\dot{R^n};B^n)_{((\cN\otimes I)(\phi))^{\otimes n}}+g(\varepsilon,|A|,n)\,,
\end{align}
where $\phi_{AR}\in\cS(AR)$ pure with $R\cong A$, and we have the fudge term
\begin{align}
g(\varepsilon,|A|,n)=\;&\frac{1}{2}\log\left(\frac{2}{(\varepsilon'/72)^2}+2\right)+ 2(|A|^2 - 1)\log(n+1) + \frac{1}{2}\log\left(\frac{2}{(\varepsilon'/24)^2}+2\right)\\
&+\frac{1}{2}\log\left(\frac{8+(\frac{\varepsilon'}{4})^2}{(\frac{\varepsilon'}{4})^2}\right) + \log\frac{4}{\varepsilon'}\quad\text{for}\quad\varepsilon' = \frac{\varepsilon}{\sqrt{2}}(n+1)^{\frac{1-|A|^2}{2}}\,.
\end{align}
\end{prop}

The mostly standard proof is included in Appendix~\ref{app:technical}. The converse is as follows.

\begin{prop}[Channel simulation converse]\label{prop:converse_channel_sim}
Let $\cN_{A\rightarrow B}$ be a quantum channel and $\varepsilon \in (0, 1]$, and $n\in\mathbb{N}$. Then, we have
\begin{align}
q^\star_{\varepsilon}(\cN^{\otimes n}) \geq \max\limits_{\phi_{AR}} \frac{1}{2}I^{\varepsilon}_{\max}(\dot{R}^n;B^n)_{((\cN\otimes I)(\phi))^{\otimes n}}
\end{align}
with $\phi_{AR}\in\cS(AR)$ is pure and $R\cong A$.
\end{prop}

\begin{proof}
Let $\cT$ be a $\{q, \varepsilon\}$-one-shot quantum channel simulation of $\cN^{\otimes n}$. Hence, $P\left(\cT,\cN^{\otimes n}\right)\leq\varepsilon$ and for any pure i.i.d.\ state $\phi_{AR}^{\otimes n}$ we have $P((\cT\otimes \mathcal{I}_{R^n})(\phi^{\otimes n}_{AR}), (\cN^{\otimes n}\otimes \mathcal{I}_{R^n})(\phi^{\otimes n}_{AR})) \leq\varepsilon$. That is, $\cT$ is a $\{q, \varepsilon\}$-one-shot quantum state splitting protocol for $\cN_{A\rightarrow B}^{\otimes n}(\phi_{A}^{\otimes n})$. Given the cost of one-shot quantum state splitting (Theorem \ref{thm:state_splitting_smooth_imax}) and that the smooth-max information is non-increasing under partial trace, we have that the minimal quantum communication cost for quantum channel simulation is lower bounded by $\frac{1}{2}I^{\varepsilon}_{\max}(\dot{R}^n;B^n)_{((\cN\otimes I)(\phi))^{\otimes n}}$. As this holds for any $\phi_{AR}\in\cS(AR)$ pure, the claim follows.
\end{proof}
\vspace{5pt}

We now have all the ingredients to prove Theorem~\ref{thm:channel_sim_mod_dev_final_result} about quantum channel simulation in the moderate deviation regime.
\vspace{5pt}

\textbf{Proof of Theorem \ref{thm:channel_sim_mod_dev_final_result}:} Recall that the minimal quantum communication cost of a simulation of $\cN^{\otimes n}$ with error $\varepsilon_n = e^{-na_n^2}$ for a strictly moderate sequence $\{a_n\}$ is denoted as $q^\star_{\varepsilon_n}(\cN^{\otimes n})$. We now combine the achievability result for one-shot channel simulation (Proposition \ref{prop:achievability_channel_sim}) with the AEP for the partially smoothed max-information (Proposition \ref{prop:smooth_imax_expansion}), and note that multiplicative factors on the error do not change the communication cost in the moderate deviation regime (Corollary \ref{cor:keps_in_Imax_expansion}). Consequently, we find for $\eta > 0$ that there exists an $n_1^\star\in\mathbb{N}$ such that for $n\geq n_1^\star$ we have
\begin{align}\label{eqn:comm_cost_maximization}
\frac{1}{n}q^\star_{\varepsilon_n}(\cN^{\otimes n}) \leq \max\limits_{\sigma} \left\{ \frac{1}{2}I(B:R)_{(\cN\otimes I)(\sigma)} + a_n\sqrt{V(B:R)_{(\cN\otimes I)(\sigma)}} + \eta a_n\right\}\,.
\end{align}
Since the set of quantum states is a compact metric space, we can apply Lemma~\ref{lem:polyanskiy_maximization} taken from \cite{polyanskiythesis} to obtain
\begin{align}
\frac{1}{n}q^\star_{\varepsilon_n}(\cN^{\otimes n}) \leq C(\cN) + \max\limits_{\sigma'\in \Pi(\cN)}a_n\sqrt{V(B:R)_{(\cN\otimes I)(\sigma')}} + o(a_n)\,.
\end{align}
Next, we combine the converse result for one-shot quantum channel simulation for i.i.d. channels (Proposition \ref{prop:converse_channel_sim}) with the AEP for the partially smoothed max-information (Proposition \ref{prop:smooth_imax_expansion}). For $\eta>0$, there exists $n_2^\star\in\mathbb{N}$ such that for $n\geq n_2^\star$, we have
\begin{align}
\frac{1}{n}q^\star_{\varepsilon_n}(\cN^{\otimes n}) &\geq  \max\limits_{\sigma_{AR}} \left\{ \frac{1}{2}I(B:R)_{(\cN\otimes I)(\sigma)} + a_n\sqrt{V(B:R)_{(\cN\otimes I)(\sigma)}} - \eta a_n\right\}\\
&\geq C(\cN) + \max\limits_{\sigma'\in \Pi(\cN)}a_n\sqrt{V(B:R)_{(\cN\otimes I)(\sigma')}} + o(a_n)\,.
\end{align}
$\qedsymbol$


\section{Entanglement-assisted quantum channel coding}\label{sec:channel_coding}

Here, we characterize the cost of one-shot entanglement-assisted channel coding in the high-error moderate deviation regime. We consider entanglement-assisted quantum communication in this section. However, the results of this section apply to entanglement-assisted classical communication as well. This is because superdense coding and teleportation establish and equivalence between these two tasks when Alice and Bob have access to shared entanglement.

Our analysis of this task relies on the close connection between high-error channel coding and low-error channel simulation. In the classical setting, this connection was explored in~\cite{bennett2002entanglement}. A similar observation has been done for the relation between source coding and the random number generation in~\cite{hayashi2008second}.

\begin{definition}[One-shot entanglement-assisted channel coding]\label{def:one_shot_channel_coding}
Consider a quantum channel $\cN_{A'\rightarrow B'}$, $\varepsilon\in[0,1]$, and quantum registers $A\cong B$ with $r = \log|A|$. A $\{r,\varepsilon\}$-one-shot quantum channel coding protocol for $\cN$ consists of
\begin{enumerate}
    \item Registers $K$, $L$ and a resource $\sigma_{KL}\in\cS(KL)$ 
    \item An encoding quantum channel $\cE_{AK\rightarrow A'}$
    \item A decoding quantum channel $\cD_{B'L\rightarrow B}$
\end{enumerate}
giving for the composite channel
\begin{align}
\widetilde{\cN}_{A\rightarrow B} = \cD\circ\cN\circ\cE\circ\cP^{\sigma} = \cD\circ\cN\circ\cE(\cdot\otimes\sigma)\quad\text{that}\quad P\left(\widetilde{\cN}_{A\rightarrow B},\cI_{A\rightarrow B}\right) \leq \varepsilon\,.
\end{align}
The maximum coding rate of $\cN_{A'\rightarrow B'}$ is defined as
\begin{align}
r^\star_{\varepsilon}(\cN) = \max \{r : \exists\text{ a}\ \{r,\varepsilon\}\text{-one-shot quantum channel coding protocol for } \cN \}\,.
\end{align}
\end{definition}

Our result in this section is the following characterization.\footnote{A complementary result for classical communication over quantum channels is given in~\cite{chubb2017moderate}.}

\begin{theorem}[Channel coding moderate deviation rate]\label{thm:channel_coding_mod_dev}
Let $\cN_{A\rightarrow B}$ be a quantum channel, $n\in\mathbb{N}$, and $\varepsilon_n = e^{-na_n^2}$ for a strictly moderate sequence $\{a_n\}$. Then, we have
\begin{align}
\frac{1}{n}r^\star_{1-\varepsilon_n}(\cN^{\otimes n}) = C(\cN) + \frac{a_n}{\sqrt{2}}\sqrt{V_{\max}(B:R)_{\cN}} + o(a_n)\,.
\end{align}
\end{theorem}

Our proof of Theorem~\ref{thm:channel_coding_mod_dev} proceeds by first upper bounding the rate of high-error channel coding with the cost of low-error channel simulation. We start by giving some notation for this section, following the definitions from the work \cite{datta2014second} on entanglement-assisted quantum channel coding for sending classical information. For any quantum channel $\cN$, let us define $P_{\rm succ}(\cN, r)$ to be the maximum probability of correctly transmitting a uniformly random message made of $r$ classical bits, where the maximization is over all encoders and decoders. We also define $P_{\rm fail}(\cN, r) = 1 - P_{\rm succ}(\cN, r)$. Finally, let $\log M^\star_{\rm ea}(\cN, \varepsilon)$ be the maximum number of bits of information that can be transmitted through the channel $\cN$ using an entanglement-assisted protocol with $P_{\rm fail}(\cN, \log M^\star_{\rm ea}(\cN, \varepsilon))$ bounded by $\varepsilon\in (0,1)$. 

\begin{lemma}[Coding converse via simulation]\label{lem:simulation_rate_bounds_coding_rate}
Let $\cN_{A'\rightarrow B'}$ be a quantum channel, $n\in\mathbb{N}$, and $\varepsilon\in (0,1]$. Then, we have
\begin{align}
r^\star_{1-\varepsilon}(\cN^{\otimes n}) \leq q^\star_{\sqrt{\varepsilon}}(\cN^{\otimes n}) + O\left(\log\frac{1}{\varepsilon}\right)\,.
\end{align}
where $q^\star_{\sqrt{\varepsilon}}(\cN^{\otimes n})$ is the minimal quantum communication cost of simulating $\cN^{\otimes n}$ as given in Definition~\ref{def:one_shot_channel_sim}.
\end{lemma}

\begin{proof}
Consider the identity channel $\cI_{A\rightarrow B}$ with $A\cong B$ and $\log|A|= q^\star_{\sqrt{\varepsilon}}(\cN^{\otimes n})$. This channel can be used for the communication required for a quantum channel simulation of $\cN^{\otimes n}$ by Definition~\ref{def:one_shot_channel_sim}. Let the simulated channel be called $\cT_{A'^n\rightarrow B'^n}$, for which we have $P(\cT, \cN^{\otimes n})\leq \sqrt{\varepsilon}$.

Now, consider an entanglement-assisted quantum channel coding protocol for $\cN^{\otimes n}$ with maximum coding rate $r^\star_{1-\varepsilon}(\cN^{\otimes n})$. There exist quantum channels $\cD, \cE$ and a preparation quantum channel $\cP^{\sigma}$ such that for $\tilde{\cN}^{n} = \cD\circ\cN^{\otimes n}\circ\cE\circ\cP^{\sigma}$, we have $P(\tilde{\cN}^{n}, I_{\bar{A}\rightarrow \bar{B}}) \leq 1-\varepsilon$ for $\bar{A}\cong\bar{B}$ and $\log|\bar{A}| = r^\star_{1-\varepsilon}(\cN^{\otimes n})$.

Next, let us now define $\tilde{\cT} = \cD\circ\cT\circ\cE\circ\cP^{\sigma}$.Using the tighter triangle inequality for the purified distance (Lemma \ref{lem:purified_dist_triangle_ineq}) and that the purified distance is non-increasing under quantum channels, we have
\begin{align}
P(\tilde{\cT}, \cI_{\bar{A}\rightarrow\bar{B}}) \leq (1-\varepsilon)\sqrt{1-\varepsilon} + \sqrt{\varepsilon}\sqrt{1 - (1-\varepsilon)^2}\leq 1 - \frac{3}{2}\varepsilon + O(\varepsilon^{3/2})\,.
\end{align}
Since we have $P(\tilde{\cT}, \cI_{\bar{A}\rightarrow\bar{B}}) \leq 1 - \frac{3}{2}\varepsilon + O(\varepsilon^{3/2})$, we can use $\tilde{\cT}$ to send $r^\star_{1-\varepsilon}(\cN^{\otimes n})$ bits with $P_{\rm fail}(\tilde{\cT}, r^\star_{1-\varepsilon}(\cN^{\otimes n}))  \leq 1 - \frac{3}{2}\varepsilon + O(\varepsilon^{3/2})$. The strong converse of the capacity of the quantum identity channel (Lemma~\ref{lem:strong_converse_ea_identity_channel}) states that if the rate exceeds the capacity, the failure probability of correctly decoding approaches $1$ exponentially, that is
\begin{align}
&P_{\rm fail}(\tilde{\cT}, r^\star_{1-\varepsilon}(\cN^{\otimes n})) \geq 1 - 2^{-(r^\star_{1-\varepsilon}(\cN^{\otimes n}) - q^\star_{\sqrt{\varepsilon}}(\cN^{\otimes n}))}\\
\implies &1 - \frac{3}{2}\varepsilon + O(\varepsilon^{3/2})\geq 1 - 2^{-(r^\star_{1-\varepsilon}(\cN^{\otimes n}) - q^\star_{\sqrt{\varepsilon}}(\cN^{\otimes n}))} \\
\implies &r^\star_{1-\varepsilon}(\cN^{\otimes n}) \leq q^\star_{\sqrt{\varepsilon}}(\cN^{\otimes n}) + O\left(\log\frac{1}{\varepsilon}\right)\,. \label{eqn:converse_channel_coding}
\end{align}
\end{proof}
\vspace{5pt}

We now have the converse statement for quantum channel coding.

\begin{prop}[Channel coding moderate deviation converse]\label{prop:converse_result_channel_coding_mod_dev}
Let $\cN_{A\rightarrow B}$ be a quantum channel, $n\in\mathbb{N}$ and $\varepsilon_n = e^{-na_n^2}$ for a strictly moderate sequence $\{a_n\}$. Then, we have that
\begin{align}
\frac{1}{n}r^\star_{1-\varepsilon_n}(\cN^{\otimes n}) \leq C(\cN) + \frac{a_n}{\sqrt{2}}\sqrt{V_{\max}(B:R)_{\cN}} + o(a_n)\,.
\end{align}
\end{prop}
\begin{proof}
We substitute $\varepsilon = e^{-na_n^2}$ for a strictly moderate sequence $\{a_n\}$ in the coding converse via simulation (Lemma \ref{lem:simulation_rate_bounds_coding_rate}) and the result then follows by our results for channel simulation in the moderate deviation regime (Theorem~\ref{thm:channel_sim_mod_dev_final_result}).
\end{proof}
\vspace{5pt}

Note that existing techniques such as the meta-converse of~\cite{matthews2014finite} yield the same bound for covariant quantum channels (as shown in Appendix~\ref{app:meta_converse_comparison}), but our converse statement as above extends that result to all channels.

Next, we show the achievability in Theorem~\ref{thm:channel_coding_mod_dev}. The main idea in the following is that if one can send classical messages using a quantum channel and entanglement-assistance, then one can also achieve entanglement-assisted quantum channel coding in the sense of Definition~\ref{def:one_shot_channel_coding} through a teleportation protocol.

\begin{lemma}[Teleportation protocol for channel coding]\label{lem:simulate_quantum_id_using_ea_noisy_ch}
Consider a quantum channel $\cN_{A'\rightarrow B'}$, $\varepsilon\in[0,1]$, and quantum registers $A\cong B$ with $\log |A| = \frac{1}{2}\log M^\star_{\rm ea}(\cN, \varepsilon)$. There exist
\begin{enumerate}
    \item registers $K$, $L$ and shared entanglement through a resource $\sigma_{KL}\in\cS(KL)$ 
    \item an encoding quantum channel $\cE_{AK\rightarrow A'}$
    \item a decoding quantum channel $\cD_{B'L\rightarrow B}$
\end{enumerate}
such that for the quantum channel
\begin{align}
\widetilde{\cN}_{A\rightarrow B} = \cD\circ\cN\circ\cE\circ\cP^{\sigma} = \cD\circ\cN\circ\cE(\cdot\otimes\sigma)\quad\text{we have}\quad P\left(\widetilde{\cN}_{A\rightarrow B},\cI_{A\rightarrow B}\right) \leq \sqrt{\varepsilon}\,.
\end{align}
\end{lemma}

\begin{proof}
Let $X$ be a classical register with $|X| = \log M^\star_{\rm ea}(\cN, \varepsilon)$. Let $\omega_{K'L'}$ be a resource state used in the entanglement-assisted classical communication protocol for $\cN_{A'\rightarrow B'}$. By the definition of $\log M^\star_{\rm ea}(\cN, \varepsilon)$, there exists an encoder $\cE'_{XK'\rightarrow A'}$ and a decoder $\cD'_{L'B'\rightarrow Y}$ where $Y\cong X$ such that we can construct the composite channel $\cD'\circ\cN\circ\cE'(\cdot\otimes\omega)$. If the  state of register $X$ is $x\in\cX$ where $x$ chosen uniformly at random from alphabet $\cX$, it holds that $y = \cD'\circ\cN\circ\cE'(x\otimes\omega)$ satisfies $\text{Pr}(y = x) = P_{\rm succ}(\cN, \log M^\star_{ea}(\cN, \varepsilon))\geq 1-\varepsilon$.

Let $\ketbra{\phi}{\phi}_{\bar{K}\bar{L}}\in\cS(\bar{K}\bar{L})$ be a $\frac{1}{2}\log M^\star_{\rm ea}(\cN, \varepsilon)$-dimensional maximally entangled state. There exists an encoder $\bar{\cE}_{A\bar{K}\rightarrow X}$ and a decoder $\bar{\cD}_{\bar{L}Y\rightarrow B}$ that is used to teleport the state in register $A$ to register $B$~\cite{bennett1993teleporting}. 

The channel $\tilde{\cN}_{A\rightarrow B}$ now combines the above elements. It uses the resource state $\sigma_{KL} = \omega_{K'L'}\otimes\ketbra{\phi}{\phi}_{\bar{K}\bar{L}}$, the encoder $\cE_{AK\rightarrow A'} = \cE'_{XK'\rightarrow A'}\circ\bar{\cE}_{A\bar{K}\rightarrow X}$ and the decoder $\cD_{B'L\rightarrow B} = \bar{\cD}_{\bar{L}Y\rightarrow B}\circ\cD'_{L'B'\rightarrow Y}$. For a fixed quantum state, the classical message that needs to be transmitted in the teleportation protocol is uniformly random and hence the teleportation of any state also occurs with success probability $p = P_{\rm succ}(\cN, \log M^\star_{ea}(\cN, \varepsilon)) \geq 1-\varepsilon$. We now apply the above protocol to the $A$ register of any pure state $\ketbra{\psi}{\psi}_{AR}$, where $R\cong A$. The output state is $p\ketbra{\psi}{\psi}_{AR} + (1-p)\ketbra{\psi^\perp}{\psi^\perp}_{AR}$. We have
\begin{align}
F(p\ketbra{\psi}{\psi}_{AR} + (1-p)\ketbra{\psi^\perp}{\psi^\perp}_{AR}, \ketbra{\psi}{\psi}_{AR})&\geq F((1-\varepsilon)\ketbra{\psi}{\psi}_{AR} + \varepsilon\ketbra{\psi^\perp}{\psi^\perp}_{AR}, \ketbra{\psi}{\psi}_{AR})\\
&= \sqrt{1-\varepsilon}\,.
\end{align}
In terms of purified distance, we have
\begin{align}
P((p\ketbra{\psi}{\psi}_{AR} + (1-p)\ketbra{\psi^\perp}{\psi^\perp}_{AR}, \ketbra{\psi}{\psi}_{AR}) \leq \sqrt{\varepsilon}\implies P\left(\tilde{\cN}_{A\rightarrow B}, \cI_{A\rightarrow B}\right) \leq \sqrt{\varepsilon}\,,
\end{align}
where the last implication follows since we can apply the protocol to any pure quantum state $\ketbra{\psi}{\psi}_{AR}$.
\end{proof}
\vspace{5pt}

Entanglement-assisted classical communication over a quantum channel is a well-studied topic. An achievability result for $\log M^\star_{\rm ea}(\cN^{\otimes n}, \varepsilon)$ is as follows.

\begin{lemma}[Entanglement-assisted classical coding achievability{~\cite[Eq.~4.64]{datta2016second}}]\label{lem:one-shot_classical_cap_datta}
Let $\cN_{A\rightarrow B}$ be a quantum channel, $\varepsilon\in(0,1)$, $g(n,\mu) = 2^{-\frac{n}{2}\left(\mu-|A| \frac{\log (n+1)}{n}\right)}$, and $0<2\delta< \varepsilon - g(n,\mu)$. Then, we have that
\begin{align}
\log M_{\rm ea}^\star(\cN^{\otimes n}, \varepsilon) \geq D_{h}^{\varepsilon -2\delta-g(n, \mu)}\left(\left(\left(\mathcal{N}\otimes\cI_R\right)\left(\psi_{AR}\right)\right)^{\otimes n} \|\left(\mathcal{N}\left(\rho_{A}\right)\right)^{\otimes n} \otimes \rho_{R}^{\otimes n}\right)-f\left(\varepsilon, \delta\right)-\log \gamma_{n, \mu}\,,
\end{align}
where $f(\varepsilon, \delta)=\log \frac{1-\varepsilon}{\delta^{2}}$, $\gamma_{n,\mu} =(n+1)^{|A|} 2^{n \mu}$, and $\psi_{AR}\in\cS(AR)$ pure.
\end{lemma}

We now have the achievability in Theorem~\ref{thm:channel_coding_mod_dev}.

\begin{prop}[Channel coding moderate deviation achievability]\label{prop:achievability_result_channel_coding_mod_dev}
Let $\cN_{A\rightarrow B}$ be a quantum channel, $n\in\mathbb{N}$, and $\varepsilon_n = e^{-na_n^2}$ for a strictly moderate sequence $\{a_n\}$. Then, we have that
\begin{align}
\frac{1}{n}r^\star_{1-\varepsilon_n}(\cN^{\otimes n}) \geq C(\cN) + \frac{a_n}{\sqrt{2}}\sqrt{V_{\max}(B:R)_{\cN}} + o(a_n)\,.
\end{align}
\end{prop}

\begin{proof}
For some strictly moderate sequence $\{a_n\}$, let us choose $\mu = \frac{2}{\ln 2}a_n^2 + |A|\log\frac{(n+1)}{n}$, $\delta = e^{-na_n^2}$ and $\varepsilon = \sqrt{1-e^{-na_n^2}}$. From this, we obtain $g(n,\mu) = e^{-na_n^2}$, $\log\gamma_{n,\mu} = 2|A|\log(n+1) + \frac{2}{\ln 2}na_n^2 = n\cdot o(a_n)$ and $f(\varepsilon, \delta) = n\cdot o(a_n)$. For some sufficiently large $n$, we then satisfy $0<2\delta< \varepsilon^2 - g(n,\mu)$. Consequently, by Lemma~\ref{lem:one-shot_classical_cap_datta} on the achievability of entanglement-assisted quantum channel coding for classical information transmission, we have for sufficiently large $n$ that 
\begin{align}
\log M_{\rm ea}^\star(\cN^{\otimes n}, \varepsilon^2) &\geq D_{h}^{\varepsilon^2 -2\delta - g(n, \mu)}\left(\left(\left(\mathcal{N}\otimes\cI_R\right)\left(\psi_{AR}\right)\right)^{\otimes n} \|\left(\mathcal{N}\left(\rho_{A}\right)\right)^{\otimes n} \otimes \rho_{R}^{\otimes n}\right)-f\left(\varepsilon, \delta\right)-\log \gamma_{n, \mu}\\
&= D_{h}^{1 - 3e^{-na_n^2}}\left(\left(\left(\mathcal{N}\otimes\cI_R\right)\left(\psi_{AR}\right)\right)^{\otimes n} \|\left(\mathcal{N}\left(\rho_{A}\right)\right)^{\otimes n} \otimes \rho_{R}^{\otimes n}\right)+ n\cdot o(a_n) \\
&\geq I(B:R)_{(\cN\otimes I)(\psi)} + \sqrt{2 V(B:R)_{(\cN\otimes I)(\psi)}} + n\cdot o(a_n)\,,
\end{align}
where the last inequality holds due to~\cite[Theorem 1]{chubb2017moderate}. The inequality above holds for any pure quantum state $\psi_{AR}$. We may choose $\psi_{AR}= \argmax_{\psi\in\Pi(\cN)} V(B:R)_{(\cN\otimes\cI_R)(\psi)}$, where $\Pi(\cN)$ is the set of capacity achieving channel inputs defined in Eq.~\eqref{eq:capacity_achieving_input_set}. By using the teleporation protocol for channel coding (Lemma \ref{lem:simulate_quantum_id_using_ea_noisy_ch}), the claim follows. 
\end{proof}
\vspace{5pt}

Combining the converse (Proposition~\ref{prop:converse_result_channel_coding_mod_dev}) and the achievability (Proposition~\ref{prop:achievability_result_channel_coding_mod_dev}), we obtain the moderate deviation analysis for channel coding (Theorem~\ref{thm:channel_coding_mod_dev}).


\section{Conclusion}\label{sec:conclusion}

In summary, we have resolved the quantum communication cost of several quantum information processing tasks in the moderate deviation regime. The task of quantum state splitting emerged as the fundamental primitive and is closely connected with quantum source coding. Our quantum state splitting results in the low-error moderate deviation regime are then used to characterize quantum channel simulation in the low-error strictly moderate deviation regime. In turn, quantum channel simulation is used to obtain a tight bound for entanglement-assisted quantum channel coding rates in the high-error strict moderate deviation regime. This extends a result that was previously only available for a restricted class of channels (as reviewed in Appendix \ref{app:meta_converse_comparison}). It would be interesting to explore if the findings in \cite{gupta2015multiplicativity} based on sandwiched R\'enyi entropies could lead to similar moderate deviation results.

Several interesting problems remain open. Our proof technique for obtaining the AEP for the partially smoothed max-information (Proposition~\ref{prop:smooth_imax_expansion}) does not work for the large error moderate deviation regime, when the error goes as $1-e^{-na_n^2}$ for a moderate sequence $\{a_n\}$. The missing technical ingredient appears to be a tighter bound on the partially smoothed max-information with the smoothed max-information (Lemma~\ref{lem:theorem4_anshu2020partially}). Hence, the problem of characterizing quantum state splitting in the high-error regime is open. In the case of quantum source coding, one can obtain a converse in the high-error regime using our quantum state splitting converse but this does not match the achievability statement in Proposition~\ref{prop:source_coding_ea_high_error} that we conjecture to be tight. As such, that problem also remains open. Finally, the polynomial multiplicative factor on the error in the de Finetti reduction (Proposition~\ref{prop:postselection}) is another obstacle that prevents us from characterizing channel simulation in the high-error strictly moderate deviation regime. It would be interesting to find an achievability proof for channel simulation without using state splitting and de Finetti reductions. This would then also yield a tighter converse bound for channel coding in the low-error strictly moderate deviation regime. 

Lastly, we remark that our quantum state splitting result does not resolve the second-order regime, where the error is a constant. This is due to a gap in the achievability and converse results in the AEP for the partially smoothed max-information (Proposition~\ref{prop:smooth_imax_expansion}). The technical ingredient required again appears to be a tighter relationship between the partially smoothed max-information and the smoothed max-information. In the moderate deviation analysis, this gap closes due to the fact that constant ($\text{poly}(n)$) multiplicative factors on the error do not affect the cost of the (strict) moderate deviation analysis as explained in Lemma~\ref{lem:moderate_seq_constant_or_polyn}. However, that no longer holds for the second-order regime.


\section{Acknowledgements}
This research is supported by the National Research Foundation, Prime Minister’s Office, Singapore and the Ministry of Education, Singapore under the Research Centres of Excellence programme. MT is also supported in part by NUS startup grants (R-263-000-E32-133 and R-263-000-E32-731).


\bibliographystyle{ultimate2.bst}
\bibliography{biblio}


\appendix

\section{Proofs one-shot quantum state splitting}\label{app:state_splitting_proofs}

Here, we prove Theorem~\ref{thm:state_splitting_smooth_imax}. We start with the convex-split lemma and define an extended state as follows: For any finite non-empty set $\Sigma$, define the quantum register $A_{\Sigma} = \bigotimes\limits_{s_i\in \Sigma} A_{s_i}$, where all $A_{s_i}$ are isomorphic. An extended state is of the form $\rho_{A_\Sigma} = \rho_{A_{s_1}}\otimes\rho_{A_{s_2}}\otimes...\otimes\rho_{A_{s_{|\Sigma|}}} \in \cS_{\leq}(A_{\Sigma})$.

\begin{lemma}[Convex-split lemma~\cite{anshu2017quantum}]\label{lem:convex_split}
For $I = \{1, 2, ...\ n\}$ consider the quantum register $B_I$ and a quantum register $R$. Then, for $i\in I$, let $\rho_{B_{i}R}\in \cS(B_{i} R)$ and $\sigma_{B_i}\in\cS(B_i)$ and define the extended state $\sigma_{B_{I\setminus i}}$ and $\tau_{B_I R} = \frac{1}{n}\sum\limits_{i=1}^{n}\rho_{B_iR}\otimes \sigma_{B_{I\setminus i}}$. For $\delta>0$ and $n$ with $\log n \geq D_{\max}(\rho_{B_iR}||\sigma_{B_i}\otimes\rho_R)  + \log\frac{1}{\delta}$, it holds that
\begin{align}
    F(\tau_{B_I R}, \sigma_{B_{I}}\otimes \rho_R) \geq \sqrt{(1 - \delta)}\,.
\end{align}
\end{lemma}

Next, we show the achievability bound for the cost of one-shot quantum state splitting. The following proposition closely follows the ideas in~\cite{anshu2017quantum}, with the sole difference that we employ the partially smoothed max-information instead of the smoothed max-information.

\begin{prop}[One-shot quantum state splitting achievability]\label{prop:state_splitting_achievability}
Consider a one-shot quantum state splitting protocol for $\rho_{AB}\in\cS(AB)$ with purification $\rho_{ABR}\in\cS(ABR)$. For $\varepsilon\in(0,1]$ with $\delta\in (0,\varepsilon]$ it holds that
\begin{align}
q^\star_{\varepsilon}(\rho_{AB}) \leq \frac{1}{2}I^{\varepsilon - \delta}_{\max}(\dot{R};B)_{\rho} + \log\frac{2}{\delta}\,.
\end{align}
\end{prop}

\begin{proof}
Let $I^{\varepsilon - \delta}_{\max}(\dot{R};B)_{\rho_{BR}} = D_{\max}(\rho'_{BR}\|\sigma_{B}\otimes\rho'_R)$, where $\rho'_{BR}\in\cB^{\varepsilon - \delta}(\rho_{BR})$, $\rho'_R = \rho_R$, and $\sigma_{B}\in\cS(B)$. For $i\in I = [n]$, let $L_i\cong B \cong A_1$ and let $\ket{\sigma_{K_iL_i}}$ be any purification of $\sigma_{L_i}$. Let the resource state be $\ket{\sigma_{(KL)_I}}\otimes\ket{\phi_{K'L'}}^{\otimes n}$, where $\ket{\phi} = \frac{1}{\sqrt{2}}(\ket{00} + \ket{11})$ is the maximally entangled qubit state. Let this resource state be shared between two parties, Alice who holds $K_{i\in I}$ and $K'^n$, and Bob who holds  $L_{i\in I}$ and $L'^n$. Excluding the maximally entangled states, the joint state at the start of the protocol is
\begin{align}
    \ket{\psi} := \ket{\rho_{AA_1R}}\otimes \ket{\sigma_{(KL)_I}}\approx_{\varepsilon - \delta} \ket{\rho'_{AA_1R}}\otimes \ket{\sigma_{(KL)_I}}:=\ket{\omega}\,,
\end{align}
where $\ket{\rho'_{AA_1R}}$ is the purification of $\rho'_{A_1R}$ on Alice's quantum register $A$ that achieves the minimal purified distance from the purification $\ket{\rho_{AA_1R}}$ of $\rho_{AR}$. Uhlmann's theorem guarantees that such a purification exists. The reduced state of $\ket{\omega}$ after tracing over Alice's quantum registers $A, A_1$ and $K_I$ is $\rho'_{R}\otimes\sigma_{L_I}$. By the convex-split lemma (Lemma~\ref{lem:convex_split}), we have that for $\delta\in(0,1]$ and $n\in \mathbb{N}$ satisfying
 \begin{align}\label{eqn:logn_lower_bound}
    \log n &\geq D_{\max}(\rho'_{BR}||\sigma_{B}\otimes\rho'_R) + \log\frac{1}{\delta^2}\,,
\end{align}
it holds that
\begin{align}
    F\left(\rho'_{R}\otimes\sigma_{L_I}, \frac{1}{n}\sum\limits_{i=1}^{n}\rho'_{RL_i} \otimes \sigma_{L_{I\setminus i}}\right) \geq \sqrt{1-\delta^2}\,.
\end{align}
A purification of $\frac{1}{n}\sum\limits_{i=1}^{n}\rho'_{RL_i} \otimes \sigma_{L_{I\setminus i}}$ is
\begin{align}
   \ket{\tau} = \frac{1}{\sqrt{n}}\sum\limits_{i=1}^{n} \ket{i}_X\otimes\ket{\rho'_{ARL_i }}\otimes\ket{\sigma_{(KL)_{I\setminus i}}}\,,
\end{align}
where Alice holds the classical register $X$. We now construct a $\{q,\varepsilon\}$-one-shot state splitting protocol of $\rho_{AB}$ assuming \eqref{eqn:logn_lower_bound} holds. 
\begin{enumerate}
    \item $\cE$ operation: There exists an isometry $U$ on Alice's quantum registers such that $F(U(\ket{\omega}),\ket{\tau})\geq \frac{1}{\sqrt{1 + \delta^2 }}$ and hence $P(U(\ket{\omega}),\ket{\tau})\leq \delta$. Moreover, since $\rho'_R = \rho_R$, there exists an isometry $V$ on $AA_1$ such that $V\ket{\psi} = \ket{\omega}$. Both these isometries exist due to Uhlmann's theorem. Alice therefore performs the operation $UV$ on $\ket{\psi}$ and discards all the $K'_{I\setminus i}$ quantum registers. 
    \item Alice and Bob use super-dense coding~\cite{bennett1992communication} to send the classical register $X$ to Bob. The state $\ket{\phi_{K'L'}}^{\otimes n}$ is used for super-dense coding and this has a quantum communication cost of $\frac{1}{2}\log n$. 
    \item $\cD$ operation: Depending on the state $\ket{i}$ of $X$, Bob applies controlled swap operations $L_i \leftrightarrow L_1$ to obtain $\ket{\rho'_{ARL_1}}$ which we relabel as $\ket{\rho'_{ABR}}$ and discards all $L_{i}$ quantum registers where $i\neq 1$. 
\end{enumerate}

Let us denote the protocol above by $\cT_\sigma$. Then, since $UV(\ket{\psi}) \approx_{\delta} \ket{\tau}$ and $\cD(\ket{\tau}) \approx_{\varepsilon - \delta}\rho_{ABR}$, we can use the triangle inequality to conclude that $\cT_\sigma(\rho_{AA_1})\approx_{\varepsilon}\rho_{AB}$. This $\varepsilon$-error one-shot quantum state splitting protocol then has a communication cost of $\frac{1}{2}\log n$ where $n\in\mathbb{N}$. Since we have 
\begin{align}
    \exp\left(D_{\max}(\rho'_{BR}||\sigma_{B}\otimes\rho'_R) + \log\frac{4}{\delta^2}\right) &\geq \left\lceil\exp\left(D_{\max}(\rho'_{BR}||\sigma_{B}\otimes\rho'_R) + \log\frac{1}{\delta^2}\right)\right\rceil\,,
\end{align}
we have that $q^\star_{\varepsilon}(\rho_{AB}) \leq D_{\max}(\rho'_{BR}||\sigma_{B}\otimes\rho'_R) + \log\frac{4}{\delta^2}$ which proves the claim.
\end{proof}
\vspace{5pt}

Next, we show the converse for one-shot quantum state splitting.

\begin{prop}[One-shot quantum state splitting converse~\cite{berta2011quantum, anshu2020partially}]\label{prop:state_splitting_converse}
Consider $\rho_{AB}\in \cS(AB)$ with purification $\rho_{ABR}$ and let $\varepsilon\in(0,1]$. The minimal quantum communication cost of one-shot state splitting of $\rho_{AB}$ satisfies
\begin{align}
q^\star_{\varepsilon}(\rho_{AB}) \geq \frac{1}{2}I^\varepsilon_{\max}(\dot{R};B)_{\rho}\,.
\end{align}
\end{prop}

\begin{proof}
We take $A_1\cong B$, initial resource state $\sigma_{KL}$ and assume the communication cost is $q = \log|Q|$. Let $\omega_{AQLR} = (\cE\otimes\mathcal{I}_{RL})(\rho_{AA_1R}\otimes\sigma_{KL})$ and $\tau_{ABR} = (\cD\otimes\mathcal{I}_{AR})(\omega_{AQLR})\approx_{\varepsilon}\rho_{ABR}$. Note that $\tau_R = \rho_R$ and thus
\begin{align}
I^{\varepsilon}_{\max}(\dot{R};B)_{\rho} &\leq I_{\max}(R;B)_{\tau}\leq  I_{\max}(R; QL)_\omega\leq I_{\max}(R;L)_\omega + 2\log|Q|\leq I_{\max}(R;L)_{\rho\otimes\sigma} + 2\log|Q|\,,
\end{align}
where we have used the data processing inequality on the max-information and the non-lockability of the max-information (Lemma~\ref{lem:lemmaA12_berta2013quantum}). Since $I_{\max}(R;L)_{\rho\otimes\sigma} = 0$ as Bob is not correlated with the reference at the start of the protocol, the result follows. 
\end{proof}


\section{Technical lemmas and missing proofs}\label{app:technical}

\subsection{Distance measures}

\begin{lemma}\label{lem:purified_dist_bound_monotone}
For $\varepsilon, \varepsilon'\in[0,1]$ with $\varepsilon^2 + \varepsilon'^2\leq 1$, the function $\varepsilon\sqrt{1 - \varepsilon'^2} + \varepsilon'\sqrt{1-\varepsilon^2}$ is monotonically increasing in both $\varepsilon$ and $\varepsilon'$. 
\end{lemma}

\begin{proof}
From the proof of tight triangle inequality for the purified distance (Lemma \ref{lem:purified_dist_triangle_ineq}), we have that $\sin^{-1}(\varepsilon) + \sin^{-1}(\varepsilon')\leq \pi/2$. Since the sine function is monotone increasing in $[0,\pi/2]$ and $\varepsilon\sqrt{1 - \varepsilon'^2} + \varepsilon'\sqrt{1-\varepsilon^2} = \sin(\sin^{-1}(\varepsilon) + \sin^{-1}(\varepsilon'))$, the result follows.
\end{proof}

\begin{lemma}\label{lem:bound_dim_R_purified_diamond_dist}
For any pair of quantum channels $\cE_{A\rightarrow B}$ and $\cF_{A\rightarrow B}$ and $\phi_{AR}\in\cS(A\otimes R)$ pure, there exists $\psi_{AR'}\in \cS(A\otimes R')$ pure such that 
\begin{align}
    P((\cE\otimes \mathcal{I}_R)(\phi) , (\cF\otimes \mathcal{I}_R)(\phi)) = P((\cE\otimes \mathcal{I}_{R'})(\psi) , (\cF\otimes \mathcal{I}_{R'})(\psi))\quad\text{with $|R'|= |A|$.}
\end{align}
\end{lemma}

\begin{proof}
If $|R|\leq |A|$, we may choose some isometry $U_{R\rightarrow R'}$ such that $\psi_{AR'} = (I_A\otimes U_{R\rightarrow R'})(\phi_{AR})(I_A\otimes U_{R\rightarrow R'})^\dagger$. Similarly, if we have that $|R|>|A|$, then we may choose $\psi_{AR'}$ to be the Schmidt decomposition of the state $\phi_{AR}$ with $|R'| = |A|$. Now there exists an isometry $V_{R'\rightarrow R}$ such that $(I_A\otimes V_{R'\rightarrow R})\psi_{AR'} = \phi_{AR}$. The purified distance is invariant under isometries acting on both arguments and the isometries here commute with the channels. This gives the desired result. 
\end{proof}

\begin{lemma}\label{lem:purified_dist_scaling}
For $\rho, \sigma \in \cS_{\leq}(A)$, let $\Tr{\lambda\rho} = \Tr{\lambda\sigma} = 1$ for some $\lambda \geq 1$. Then, we have
\begin{align}
P(\rho,\sigma) \leq P(\lambda\rho,\lambda\sigma) &\leq \sqrt{2\lambda}P(\rho,\sigma)\,.
\end{align}
\end{lemma}

\begin{proof}
We use the generalized fidelity for sub-normalized states. We have
\begin{align}
F(\lambda\rho,\lambda\sigma) &= \lambda\|\sqrt{\rho}\sqrt{\sigma}\|_1= \lambda\left(\bar{F}(\rho,\sigma) - \left(1-\frac{1}{\lambda}\right)\right)= 1 - \lambda(1 - \bar{F}(\rho,\sigma))\,.
\end{align}
Rewriting, we have
\begin{align}
\bar{F}(\rho,\sigma) = \frac{1}{\lambda}F(\lambda\rho,\lambda\sigma) + 1 - \frac{1}{\lambda}\,.
\end{align}
The right hand side is a convex combination of $F(\lambda\rho,\lambda\sigma)$ and $1$ and we can hence lower bound it by $F(\lambda\rho,\lambda\sigma)$. Switching to purified distance, we have the first inequality. We also have
\begin{align}
P^2(\lambda\rho,\lambda\sigma) = 1 - (1 - \lambda(1 - \bar{F}(\rho,\sigma)))^2&= 2\lambda(1 - \bar{F}(\rho,\sigma)) - \lambda^2(1 - \bar{F}(\rho,\sigma))^2 \\
&\leq 2\lambda(1 - \bar{F}(\rho,\sigma))\\
&\leq 2\lambda(1 - \bar{F}^2(\rho,\sigma))\\
&= 2\lambda P^2(\rho,\sigma)\,.
\end{align}
\end{proof}

\begin{lemma}\label{lem:fidelity_cq_states}
For classical-quantum states $\rho = \frac{1}{d}\sum\limits_{i=1}^d \ket{i}\bra{i}\otimes \rho_i$ and $\sigma = \frac{1}{d}\sum\limits_{i=1}^d\ket{i}\bra{i}\otimes\sigma_i$, where $\{\ket{i}\}$ form an orthonormal basis, it holds that
\begin{align}
F(\rho,\sigma) = \frac{1}{d}\sum\limits_{i=1}^d F(\rho_i,\sigma_i)\,.
\end{align}
\end{lemma}

\begin{proof}
This follows from the fact that for classical-quantum states $\tau = \sum_i p(i)\ket{i}\bra{i}\otimes \tau_i$ and $\omega = \sum_i q(i)\ket{i}\bra{i}\otimes \omega_i$ where $p(i), q(i)$ are probability vectors and $\tau_i$ and $\omega_i$ are quantum states, one has $F(\tau,\omega) = \sum_i\sqrt{p(i)q(i)}F(\tau_i,\omega_i)$. 
\end{proof}

\begin{lemma}[Quasi-convexity of purified distance~{\cite[Eq. (3.54)]{tomamichel2015quantum}}]\label{lem:quasi_convexity_purified_dist}
For $\lambda\in[0,1]$ $i\in\{1,2\}$ and $\rho_i, \tau_i\in \cS_{\leq}(A)$, it holds that
\begin{align}
P\left(\lambda \rho_{1}+(1-\lambda) \rho_{2}, \lambda \tau_{1}+(1-\lambda) \tau_{2}\right) \leq \max_i P\left(\rho_{i}, \tau_{i}\right)\,.
\end{align}
\end{lemma}


\subsection{Entropic quantities}

\begin{lemma}[Mutual information variance {\cite[Corollary III.5]{dupuis2019entropy}}]\label{lem:bounded_mutual_info_variance}
For $\rho_{AB}\in\cS(AB)$, we have that
\begin{align}
V(\rho_{AB}\|\rho_A\otimes\rho_B) \leq 4 \log ^{2}\left(2 d_{A}+1\right)\,.
\end{align}
\end{lemma}

\begin{lemma}[{\cite[Proposition 4.7]{datta2014second}}]\label{lem:prop4.7_datta2014second}
Let $0<\varepsilon<1, \delta, \eta>0$, and $\rho, \sigma\in\cP(A)$ with $\Tr{\rho}\leq 1$. Then, we have
\begin{align}
D_{h}^{\varepsilon-\delta}(\rho \| \sigma)+\log \delta \leq \underline{D}_{s}^{\varepsilon}(\rho \| \sigma)\,.
\end{align}
\end{lemma}

\begin{lemma}[{\cite[Theorem 4]{anshu2019minimax}}]\label{lem:thm4_anshu2019minimax}
Let $\rho \in \mathcal{S}(A), \sigma \in \mathcal{P}(A)$, $\varepsilon \in(0,1)$, and $\delta \in\left(0,1-\varepsilon^{2}\right)$. Then, we have that
\begin{align}
    D_{h}^{1-\varepsilon}(\rho \| \sigma) \geq D_{\max }^{\sqrt{\varepsilon}, P}(\rho \| \sigma)-\log \frac{1}{1-\varepsilon} \geq D_{h}^{1-\varepsilon-\delta}(\rho \| \sigma)-\log \frac{4}{\delta^{2}}\, .
\end{align}
\end{lemma}

\begin{lemma}[Duality of Rényi relative entropies~{\cite[Lemma 6]{hayashi2016correlation}}]\label{lem:duality_renyi_relative_entropies}
Let $\rho_{A B} \in \mathcal{S}(A B)$ and $\tau_{A} \geq 0$ such that $\tau_{A} \gg \rho_{A}$. Then, for any purification $\rho_{A B C}$ of $\rho_{A B}$, we have
\begin{align}
\widetilde{I}_{\alpha}\left(\rho_{A B} \| \tau_{A}\right) &=-\widetilde{I}_{\beta}\left(\rho_{A C} \| \tau_{A}^{-1}\right) & &\quad \text{for} \quad \alpha, \beta \in[1 / 2, \infty), \quad \alpha^{-1}+\beta^{-1}=2\\
D_{\alpha}\left(\rho_{A B} \| \tau_{A} \otimes \rho_{B}\right)&=-D_{\beta}\left(\rho_{A C} \| \tau_{A}^{-1} \otimes \rho_{C}\right) & &\quad\text {for} \quad \alpha, \beta \in[0,2], \quad \alpha+\beta=2 \\
V\left(\rho_{AC} \| \tau_{A}^{-1}\otimes\rho_{C}\right) &= V\left(\rho_{AB} \| \tau_{A}\otimes\rho_{B}\right)\,,
\end{align}
where the inverse is taken on the support of $\tau_A$.
\end{lemma}

\begin{lemma}[Non-lockability of max-information~{\cite[Lemma A.12]{berta2013quantum}}]\label{lem:lemmaA12_berta2013quantum}
Let $\rho_{ABC}\in \mathcal{S}(ABC)$ and $\varepsilon\geq 0$. Then, we have that
\begin{align}
I_{\max }^{\varepsilon}(A;BC)_{\rho} \leq I_{\max }^{\varepsilon}(A;B)_{\rho}+2\log |C|\,.
\end{align}
\end{lemma}

\begin{lemma}[Partially smoothed max-information bound~{\cite[Theorem 2]{anshu2020partially}}]\label{lem:theorem4_anshu2020partially}
Let $\rho_{A B} \in \mathcal{S}(AB)$ and $0\leq 2\varepsilon+\delta \leq 1$ with $\delta>0$. Then, we have 
\begin{align}
    I_{\max }^{2 \varepsilon+\delta, P}(\dot{A} ; B)_{\rho} \leq I_{\max }^{\varepsilon, P}(A ; B)_{\rho}+\log \frac{8+\delta^{2}}{\delta^{2}}\, .
\end{align}
\end{lemma}


\subsection{De Finetti and more}

\begin{definition}[De Finetti state]\label{def:definetti_state}
For $\sigma \in \mathcal{S}(A)$ and $\mu(.)$ a probability measure on $\mathcal{S}(A)$, $\zeta_{H^n}=\int \sigma^{\otimes n} \mu(\sigma) \in \mathcal{S}\left(A^{n}\right)$ is called a de Finetti state. 
\end{definition}

\begin{lemma}[Post-selection technique~\cite{christandl2009postselection}]\label{lem:postselection_technique_symmetric_state}
For a $\rho_{A^nR^n}\in\cS(A^nR^n)$ with support on the symmetric subspace $\text{Sym}^n(A\otimes R)$ with $R\cong A$, there exists a completely positive trace-non-increasing map $\cG_{R'\rightarrow \mathbb{C}}$ such that
\begin{align}
    \rho_{A^nR^n} = g_{n,|A|}(\mathcal{I}_{A^nR^n}\otimes\cG)(\zeta_{A^nR^nR'})\,,
\end{align}
where $g_{n,d} = \left(\begin{array}{c} n+d^{2}-1 \\ n \end{array}\right) \leq(n+1)^{d^{2}-1}$, $\zeta_{A^n R^n} = \int \sigma_{A R}^{\otimes n}\ d\left(\sigma_{A R}\right)$ with $\sigma_{A R}=\vert\sigma\rangle\langle\sigma\vert_{A R} \in \mathcal{S}\left(AR\right)$, $A \cong R$ and $d(.)$ is the measure on the normalized pure states on $AR$ induced by the Haar measure on the unitary group acting on $AR$ normalized to $\int d(.)=1$. Hence, $\zeta_{A^nR^nR'} \in \cS(A^nR^nR')$ is the purification of $\zeta_{A^nR^n}$ and we can assume without loss of generality that $\left|R^{\prime}\right|\leq(n+1)^{|A|^{2}-1}$.
\end{lemma}

The following proposition is as \cite[Theorem 1]{christandl2009postselection}, but we changed the distance measure from the diamond distance to the channel purified distance.

\begin{prop}[De Finetti reduction]\label{prop:postselection}
Let $\mathcal{E}^n_{A^n\rightarrow B^n},\mathcal{F}^n_{A^n\rightarrow B^n}$ be quantum channels and $\varepsilon>0$. If $\cE^n$ and $\cF^n$ are permutation covariant, i.e., for any permutation $\pi$ on $n$ quantum registers, one has $\mathcal{E}^n \circ \pi_{A^n}=\pi_{B^n} \circ\mathcal{E}^n$ and $\mathcal{F}^n \circ \pi_{A^n}=\pi_{B^n} \circ\mathcal{F}^n$, where $\pi_{A^n}$ ($\pi_{B^n}$) denotes the channel that applies the permutation $\pi$ to the $A^n$ ($B^n$) registers. Then 
\begin{align}
P\left(\mathcal{E}^n,\mathcal{F}^n\right)\leq\sqrt{2}(n+1)^{\frac{\left(|A|^{2}-1\right)}{2}} P\left(\left(\mathcal{E}^n \otimes \mathcal{I}_{R^n R^{\prime}}\right) (\zeta_{A^n R^n R^{\prime}}) , \left(\mathcal{F}^n \otimes \mathcal{I}_{R^n R^{\prime}}\right) (\zeta_{A^n R^n R^{\prime}})\right)\,,
\end{align}  
where $\zeta_{A^n R^nR'}$ is the purification of the de Finetti state defined in Lemma \ref{lem:postselection_technique_symmetric_state}.
\end{prop}

\begin{proof}
We wish to show that 
\begin{align}
&P\left(\left(\mathcal{E}^{n} \otimes \mathcal{I}_{R^n}\right) (\rho_{A^n R^n}) , \left(\mathcal{F}^{n} \otimes \mathcal{I}_{R^n}\right) (\rho_{A^n R^n})\right) \\
&\leq \sqrt{2g_{n, |A|}}\cdot P\left(\left(\mathcal{E}^{n} \otimes \mathcal{I}_{R^n R^{\prime}}\right) (\zeta_{A^n R^n R^{\prime}}) , \left(\mathcal{F}^{n} \otimes \mathcal{I}_{R^n R^{\prime}}\right) (\zeta_{A^n R^n R^{\prime}})\right)\,,
\end{align}
where $R\cong A$. First we may assume that $\rho_{A^nR^n}$ has support on $\text{Sym}^n(A\otimes R)$. To see this, define the quantum state
\begin{align}
\bar{\rho}_{A^nR^nR'} = \frac{1}{n!}\sum_\pi (\pi_{A^n}\otimes \mathcal{I}_{R^n})(\rho_{A^nR^n})\otimes\ket{\pi}\bra{\pi}_{R'}\,,
\end{align}
where $\pi$ is a permutation operation on the $n$ quantum registers, the summation is over all permutations on $A^n$ and $\{\ket{\pi}\}$ is a basis of $R'$. We now have
\begin{align}
&F\left(\left(\mathcal{E}^{n} \otimes \mathcal{I}_{R^n}\right) (\rho_{A^n R^n}) , \left(\mathcal{F}^{n} \otimes \mathcal{I}_{R^n}\right) (\rho_{A^n R^n})\right)\\ 
&= \frac{1}{n!}\sum_{\pi} F\left(\left(\pi_{B^n}\circ\mathcal{E}^{n} \otimes \mathcal{I}_{R^n}\right) (\rho_{A^n R^n}) , \left(\pi_{B^n}\circ\mathcal{F}^{n} \otimes \mathcal{I}_{R^n}\right) (\rho_{A^n R^n})\right)\\
&= \frac{1}{n!}\sum_{\pi} F\left(\left(\mathcal{E}^{n} \otimes \mathcal{I}_{R^n}\right)(\pi_{A^n}\otimes \mathcal{I}_{R^n})(\rho_{A^n R^n}) , \left(\mathcal{F}^{n} \otimes \mathcal{I}_{R^n}\right) (\pi_{A^n}\otimes \mathcal{I}_{R^n})(\rho_{A^n R^n})\right)\\
&=F\left(\left(\mathcal{E}^{n} \otimes \mathcal{I}_{R^nR'}\right) (\bar{\rho}_{A^n R^n R'}) , \left(\mathcal{F}^{n} \otimes \mathcal{I}_{R^nR'}\right) (\bar{\rho}_{A^n R^n R'})\right)\,,
\end{align}
where the first equality holds due to the invariance of the fidelity under isometries, the second equality holds due to the permutation covariance of $\cE^n$ and $\cF^n$, the third equality is due to Lemma \ref{lem:fidelity_cq_states} on the fidelity of classical-quantum states and the linearity of the channels $\cE^n$ and $\cF^n$. 

The reduced state $\bar{\rho}_{A^n}$ is permutation invariant and therefore has a permutation invariant purification $\hat{\rho}_{A^n A'{^n}}$ with $A'\cong A$~\cite[Lemma 4.2.2]{renner2005security}. Since all purifications are equivalent up to isometries on the purifying quantum register, there exists a quantum channel $\cM_{A'^n\rightarrow R^nR'}$ such that $\bar{\rho}_{A^nR^nR'} = (\mathcal{I}_{A^n}\otimes \cM)(\hat{\rho}_{A^n A'{^n}})$. Switching to the purified distance, we have
\begin{align}
    &P\left(\left(\mathcal{E}^{n} \otimes \mathcal{I}_{R^n}\right)(\rho_{A^n R^n}) , \left(\mathcal{F}^{n} \otimes \mathcal{I}_{R^n}\right)(\rho_{A^n R^n})\right)\\
    &=P\left(\left(\mathcal{E}^{n} \otimes \mathcal{I}_{R^nR'}\right)\circ(\mathcal{I}_{A^n}\otimes \cM)(\hat{\rho}_{A^n A'{^n}}) , \left(\mathcal{F}^{n} \otimes \mathcal{I}_{R^nR'}\right)\circ (\mathcal{I}_{A^n}\otimes \cM)(\hat{\rho}_{A^n A'{^n}})\right) \\
    &\leq P\left(\left(\mathcal{E}^{n} \otimes \mathcal{I}_{A'^n}\right)(\hat{\rho}_{A^n A'{^n}}) , \left(\mathcal{F}^{n} \otimes \mathcal{I}_{A'^n}\right)(\hat{\rho}_{A^n A'{^n}})\right)\,,
\end{align}
where the inequality holds because the purified distance is non-increasing under quantum channels. Now, we have
\begin{align}
    &P\left(\left(\mathcal{E}^{n} \otimes \mathcal{I}_{A'^n}\right)(\hat{\rho}_{A^n A'{^n}}) , \left(\mathcal{F}^{n} \otimes \mathcal{I}_{A'^n}\right)(\hat{\rho}_{A^n A'{^n}})\right) \\
    &= P\left(\left(\mathcal{E}^{n} \otimes \mathcal{I}_{A'^n}\right)\circ g_{n,|A|}(\mathcal{I}_{A^nA'^n}\otimes\cG)(\zeta_{A^nA'^nR'}) , \left(\mathcal{F}^{n} \otimes \mathcal{I}_{A'^n}\right)\circ g_{n,|A|}(I_{A^nA'^n}\otimes\cG)(\zeta_{A^nA'^nR'})\right) \\
    &\leq \sqrt{2g_{n,|A|}}P\left(\left(\mathcal{E}^{n} \otimes \mathcal{I}_{A'^n}\right)\circ(\mathcal{I}_{A^nA'^n}\otimes\cG)(\zeta_{A^nA'^nR'}) , \left(\mathcal{F}^{n} \otimes \mathcal{I}_{A'^n}\right)\circ(\mathcal{I}_{A^nA'^n}\otimes\cG)(\zeta_{A^nA'^nR'})\right) \\
    &\leq \sqrt{2g_{n,|A|}}P\left(\left(\mathcal{E}^{n} \otimes \mathcal{I}_{A'^n}\right)(\zeta_{A^nA'^nR'}) , \left(\mathcal{F}^{n} \otimes \mathcal{I}_{A'^n}\right)(\zeta_{A^nA'^nR'})\right)\,,
\end{align}
where the first equality follows from Lemma \ref{lem:postselection_technique_symmetric_state}, the first inequality is due to Lemma \ref{lem:purified_dist_scaling} and the second inequality is because the purified distance cannot increase under completely positive trace non-increasing maps.
\end{proof}
\vspace{5mm}

\textbf{Proof of Lemma~\ref{lem:perm_inv_state_splitting}:} We construct $\bar{\cT}^{\omega}$ from $\cT^{\sigma}$ as follows. The resource state is $\omega = \sigma_{KL}\otimes\sigma'_{K'L'}$, where
\begin{align}
    \sigma'_{K'L'} = \frac{1}{n!}\sum_{\pi}\ket{\pi\pi}\bra{\pi\pi}_{K'L'}\,.
\end{align}
The encoders and decoders are modified from $\cE$ and $\cD$ respectively to $\cE' = \frac{1}{n!}\sum_{\pi} \cE\circ(\pi_{A^n}\otimes\pi_{A_1^n}) \otimes \bra{\pi}\cdot\ket{\pi}_{K'}$ and $\cD' = \frac{1}{n!}\sum_{\pi} (\pi^{-1}_{A^n}\otimes\pi_{B^n}^{-1})\circ\cD \otimes \bra{\pi}\cdot\ket{\pi}_{L'}$. The following properties are clear:
\begin{enumerate}
    \item $\bar{\cT}^{\omega} = \frac{1}{n!}\sum_{\pi} (\pi_{A^n}^{-1}\otimes\pi_{B^n}^{-1})\circ \cT^{\sigma}\circ (\pi_{A^n}\otimes\pi_{A_1^n})$
    \item $\bar{\cT}^{\omega} = \cD'\circ\cE'\circ\cP^\omega$ satisfies $\bar{\cT}^{\omega}\circ(\bar{\pi}_{A^n}\otimes\bar{\pi}_{A_1^n}) = (\bar{\pi}_{A^n}\otimes\bar{\pi}_{B^n})\circ \bar{\cT}^{\omega} $
    \item The communication cost $q$ associated with $\cT^\sigma$ is the same as that of $\bar{\cT}^\omega$. 
\end{enumerate}

Since $\rho_{A^nB^n}$ is permutation invariant, it has a permutation invariant purification $\rho_{A^nB^nR^n}$ i.e., $(\pi_{A^n}\otimes\pi_{B^n}\otimes\pi_{R^n})\rho_{A^nB^nR^n} = \rho_{A^nB^nR^n}$ for any permutation $\pi$~\cite[Lemma 4.2.2]{renner2005security}. We can, without loss of generality, consider this purification due to Remark~\ref{rmk:equivalence_of_purifications_state_splitting}.
We have that
\begin{align}
    &P\left(\rho_{A^nB^nR^n} , (\bar{\cT}^{\omega}\otimes \mathcal{I}_{R^n})(\rho_{A^nA_1^nR^n})\right) \nonumber \\
    &=P\left(\rho_{A^nB^nR^n} , \left(\frac{1}{n!}\sum\limits_{\pi}(\pi_{A^n}^{-1}\otimes\pi_{B^n}^{-1})\circ\cT^\sigma\circ(\pi_{A^n}\otimes\pi_{A_1^n})\otimes \mathcal{I}_{R^n}\right)(\rho_{A^nA_1^nR^n})\right) \\
    &= P\left(\rho_{A^nB^nR^n} , \left(\frac{1}{n!}\sum\limits_{\pi}(\pi_{A^n}^{-1}\otimes\pi_{B^n}^{-1})\circ\cT^\sigma \otimes \pi^{-1}_{R^n}\right)(\rho_{A^nA_1^nR^n})\right) \\ 
    &= P\left(\rho_{A^nB^nR^n} , \left(\frac{1}{n!}\sum\limits_{\pi}(\pi_{A^n}^{-1}\otimes\pi_{B^n}^{-1}\otimes\pi^{-1}_{R^n})\circ\cT^\sigma \otimes \cI_{R^n}\right)(\rho_{A^nA_1^nR^n})\right) \\ 
    &\leq \max_{\pi}P\left(\rho_{A^nB^nR^n} , \left((\pi_{A^n}^{-1}\otimes\pi_{B^n}^{-1}\otimes\pi^{-1}_{R^n})\circ\cT^\sigma \otimes \cI_{R^n}\right)(\rho_{A^nA_1^nR^n})\right) \\ 
    &= \max_{\pi}P\left((\pi_{A^n}\otimes\pi_{B^n}\otimes\pi_{R^n})(\rho_{A^nB^nR^n}) , \left(\cT^\sigma \otimes \cI_{R^n}\right)(\rho_{A^nA_1^nR^n})\right) \\ 
    &=P\left(\rho_{A^nB^nR^n} , \left(\cT^\sigma \otimes \mathcal{I}_{R^n}\right)(\rho_{A^nA_1^nR^n})\right)\,,
\end{align}
where the second equality uses the fact that the $\rho$ is permutation invariant, the third equality is because the state splitting protocol commutes with operations on the $R^n$ quantum registers, the first inequality uses the quasi-convexity of the purified distance (Lemma \ref{lem:quasi_convexity_purified_dist}), the fourth equality uses the invariance of the purified distance under isometries and the last equality uses that $\rho$ is permutation invariant. 
$\qedsymbol$


\subsection{Miscellaneous}

\begin{lemma}[{\cite[Lemma 48]{polyanskiythesis}}]\label{lem:polyanskiy_maximization}
Let $D$ be a compact metric space and $\{a_n\}$ be a moderate sequence. For any continuous functions $f:D\rightarrow\mathbb{R}$ and $g:D\rightarrow\mathbb{R}$ we have
\begin{align}
    \max _{x \in D}\Big\{f(x)+ a_ng(x)\Big\}= f^{*}+a_ng^{*}+o(a_n)
\end{align}
for $f^{*} =\max _{x \in D} f(x)$ and $g^{*}=\sup _{\left\{x: f(x)=f^{*}\right\}} g(x)$.
\end{lemma}

\begin{lemma}[Strong converse for quantum identity channel]\label{lem:strong_converse_ea_identity_channel}
Consider a quantum identity channel $\cI_{A\rightarrow B}$ where $A\cong B$. The average success probability of correctly transmitting $2^{R}$ classical messages over $\cI$ with entanglement-assistance is bounded as  
\begin{align}
    P_{\text{succ}}(\cI_{A\rightarrow B}, R) \leq 2^{-(R - 2\log|A|)}\, .
\end{align}
\end{lemma}
\begin{proof}
Let Alice send a uniformly random message $m$ from the set $M$ where $|M| = 2^R$ with average success probability $P_{\rm succ}$ using the channel $\cI_{A\rightarrow B}$. To do so, Alice and Bob share an entangled resource $\sigma_{KL}$. For each message $m\in M$, the encoder can be chosen without loss of generality as $\cE^m_{K\rightarrow A}$ taken from the set $\{\cE^m_{K\rightarrow A}\}$. Let $\rho_{A} =\frac{1}{|M|}\sum_{m\in M}\cE^m_{K\rightarrow A}(\sigma_K)$. By the meta-converse argument of~\cite{matthews2014finite} (see also~\cite[Proposition 1]{gupta2015multiplicativity}), there exists a two-element POVM $\{T_{RB}, I - T_{RB}\}$ such that for any purification $\psi_{RA}$ of $\rho_{A}$ and any state $\omega_{A}$, we have
\begin{align}
P_{\rm succ} &= \Tr{T_{RB}\cI_{A\rightarrow B}(\psi_{RA})} \\
\frac{1}{|M|}&= \Tr{T_{RB}\cI_{A\rightarrow B}(\psi_{R}\otimes\omega_{A})} 
\end{align}
Finally, define the state $\rho_p = p\ketbra{0}{0} + (1-p)\ketbra{1}{1}$ for any $p\in[0,1]$. We now have the following chain of inequalities
\begin{align}
    \log\frac{P_{\rm succ}}{\frac{1}{|M|}}&\leq D_{\max}\left(\rho_{P_{\rm succ}}\|\rho_{\frac{1}{|M|}}\right)\\
    &\leq D_{\max}\left(\psi_{RA}\|\psi_{R}\otimes\frac{I_{A}}{|A|}\right) \\
    &\leq 2\log|A|, 
\end{align}
where the first inequality is due to the definition of the max-relative entropy, the second inequality is due to the monotonicity of the max-relative entropy under the test $\{T_{RB}, I - T_{RB}\}$ and the choice of $\omega_{A} = \frac{I_{A}}{|A|}$, and the final inequality is due to~\cite[Lemma B.6]{berta2013quantum}.
\end{proof}
\vspace{5mm}

\textbf{Proof of Proposition~\ref{prop:achievability_channel_sim}:}
Let $\zeta_{A^n}$ be the de Finetti state as given in Lemma \ref{lem:postselection_technique_symmetric_state}. Let $U^{\otimes n}_{A\rightarrow A_1E}$ be the Stinespring dilation of the $n$-fold tensor product of the channel $\cN^{\otimes n}_{A\rightarrow A_1}$ and denote $U^{\otimes n}_{A\rightarrow A_1E}(\zeta_{A^n}) = \Phi_{A_1^nE^n}$. Since $\zeta_{A^n}$ is permutation invariant, we have that $\Phi_{A_1^nE^n}$ is also permutation invariant. For any $\varepsilon'\in (0,1]$ and $\delta\in (0, \varepsilon')$, let $\cT_{n}$ be $\{q_n, \varepsilon'\}$-one-shot state splitting protocol of $\Phi_{A_1^nE^n}$. By Lemma~\ref{lem:perm_inv_state_splitting}, we can assume that $\cT_n$ is permutation covariant. Let $\cT'_n = \text{Tr}_{E^n}\circ\cT_n\circ U^{\otimes n}$. We make the following observations about $\cT'_n$:
\begin{enumerate}
    \item $\cT'_n$ is permutation covariant
    \item For any purification $\zeta_{A^nR}$ of $\zeta_{A^n}$, we have $P\left(\left(\mathcal{N}_{A\rightarrow B}^{\otimes n}\otimes \mathcal{I}_{R}\right)(\zeta_{A^n R}), \left(\cT'_{n} \otimes \mathcal{I}_{R}\right)(\zeta_{A^n R })\right) \leq \varepsilon'$
\end{enumerate}
By choosing $\varepsilon' \leq \frac{\varepsilon}{\sqrt{2}}(n+1)^{\frac{1-|A|^2}{2}}$, we obtain that $P\left(\mathcal{N}_{A\rightarrow B}^{\otimes n},\cT'_n\right) \leq \varepsilon$ through the de Finetti reduction (Proposition \ref{prop:postselection}). $\cT'_n$ is therefore a $\{q_n, \varepsilon\}$-one-shot simulation protocol of $\cN^{\otimes n}_{A\rightarrow B}$.

It now remains to pick a purification of $\zeta_{A^n}$ to calculate $q_n$. By \cite[Corollary D.6]{berta2011quantum}, we may choose the extension $\zeta_{A^nR^n} = \sum\limits_{i\in I} p_i(\omega^i_{AR})^{\otimes n}$ where $I = \left\{1,2, \ldots,(n+1)^{2|A||R|-2}\right\}$ and $\omega^i_{AR}$ are pure states. Since this state is permutation invariant, we may choose a permutation invariant purification $\zeta_{A^nR^nR'}$ with $|R'|\leq (n+1)^{|A|^2-1}$. By the one-shot characterization of quantum state splitting (Theorem~\ref{thm:state_splitting_smooth_imax}) and by choosing $\delta = \frac{\varepsilon'}{2}$, we have
\begin{align}
    q^\star_{\varepsilon}(\cN^{\otimes n}_{A\rightarrow B}) &\leq q^\star_{\varepsilon'}(\Phi_{B^nE^n}) \\
    &\leq \frac{1}{2}I^{\frac{\varepsilon'}{2}}_{\max}(\dot{R^nR'};B^n)_{\Phi} + \log\frac{4}{\varepsilon'}\\
    &\leq \frac{1}{2}I^{\frac{\varepsilon'}{8}}_{\max}(R^nR';B^n)_{\Phi}+  \frac{1}{2}\log\frac{8+(\frac{\varepsilon'}{4})^2}{(\frac{\varepsilon'}{4})^2} + \log\frac{4}{\varepsilon'}\\
    &\leq \frac{1}{2}I^{\frac{\varepsilon'}{24}}_{\max}(B^n;R^nR')_{\Phi}+ \frac{1}{2}\log\left(\frac{2}{(\varepsilon'/24)^2}+2\right) + \frac{1}{2}\log\frac{8+(\frac{\varepsilon'}{4})^2}{(\frac{\varepsilon'}{4})^2} + \log\frac{4}{\varepsilon'}\, ,
\end{align}
where the second inequality holds due to the cost of one-shot state splitting (Theorem \ref{thm:state_splitting_smooth_imax}), the third inequality follows due to the bound on the partially smoothed max-information using the smoothed max-information (Lemma~\ref{lem:theorem4_anshu2020partially}) and the fourth inequality follows due to the approximate symmetry of the smooth max-information \cite[Lemma 2.11]{berta2013quantum}). We now have
\begin{align}
   \frac{1}{2}I^{\frac{\varepsilon'}{24}}_{\max}(B^n;R^nR')_{\Phi}
    &\leq \frac{1}{2}I^{\frac{\varepsilon'}{24}}_{\max}(B^n;R^n)_{\Phi}+ \log|R'| \\
    &\leq \max\limits_{i\in I} \frac{1}{2}I^{\frac{\varepsilon'}{24}}_{\max}(B^n;R^n)_{((\cN\otimes I)(\omega^i))^{\otimes n}}+ \frac{1}{2}\log|I| + \log|R'| \\
    &\leq  \max\limits_{\tau_{AR}} \frac{1}{2}I^{\frac{\varepsilon'}{24}}_{\max}(B^n;R^n)_{((\cN\otimes I)(\tau))^{\otimes n}}+ 2(|A|^2 - 1)\log(n+1)\,,
\end{align}
where the first inequality holds due to a bound on the increase of the max-information due to an additional quantum register (Lemma B.9 of~\cite{berta2011quantum}), the second inequality follows by the quasi-convexity of the smooth max-entropy (Lemma B.18 of~\cite{berta2011quantum}) and the third inequality holds due to the bound on the dimension of the purification of the de Finetti state in the de Finetti reduction (Proposition \ref{prop:postselection}) and the bound on $|I|$. We use again the approximate symmetry of the smoothed max-information and bound the the smoothed-max information by the partially smoothed max-information to obtain
\begin{align}
    \max\limits_{\tau_{AR}} \frac{1}{2}I^{\frac{\varepsilon'}{24}}_{\max}(B^n;R^n)_{((\cN\otimes I)(\tau))^{\otimes n}} \leq \max\limits_{\tau_{AR}} \frac{1}{2}I^{\frac{\varepsilon'}{72}}_{\max}(\dot{R^n};B^n)_{((\cN\otimes I)(\tau))^{\otimes n}}+ \frac{1}{2}\log\left(\frac{2}{(\varepsilon'/72)^2}+2\right)\,.
\end{align}
$\qedsymbol$


\section{Meta-converse channel coding}\label{app:meta_converse_comparison}

Here, we compare the results from Section \ref{sec:channel_coding} with results obtained by a meta-converse argument for a certain class of quantum channels. Our meta-converse below adapts the argument in~\cite{matthews2014finite} to the channel purified distance.

\begin{prop}[Meta-converse channel coding]\label{prop:meta_converse_bound}
Let $\cN_{A'\rightarrow B'}$ be a quantum channel, registers $A, B$ and $R$ all be isomorphic, and $\varepsilon\in(0,1]$. Suppose there exists some $\widetilde{\cN}_{A\rightarrow B}$ per Definition \ref{def:one_shot_channel_coding} such that $P(\widetilde{\cN},\cI_R) \leq \varepsilon$. Then, we have that
\begin{align}
r^\star_{\varepsilon}(\cN) \leq \max_{\rho_{A'R}}\min_{\sigma_{B'}} \frac{1}{2}D^{\varepsilon^2}_h(\left(\cN\otimes\cI_R\right)(\rho_{A'R})\| \sigma_{B'}\otimes\rho_R)\,.
\end{align}
\end{prop}

\begin{proof}
We are given that $P(\tilde{\cN} , \cI_R) \leq \varepsilon$. Let $\ket{\phi} = \frac{1}{\sqrt{|R|}}\sum_{i=1}^{|R|}\ket{i}\ket{i}$ be the $|R|$-dimensional maximally entangled state on the appropriate quantum registers. It holds that 
\begin{align}
&P\left(\left(\tilde{\cN}\otimes\cI_R\right)(\ketbra{\phi}{\phi}_{AR}), \ketbra{\phi}{\phi}_{BR}\right)\leq \varepsilon\\
\implies &F\left(\left(\tilde{\cN}\otimes\cI_R\right)(\ketbra{\phi}{\phi}_{AR}), \ketbra{\phi}{\phi}_{BR}\right) \geq \sqrt{1-\varepsilon^2}\\
\implies &\Tr{\ketbra{\phi}{\phi}_{BR}\ \left(\tilde{\cN}\otimes\cI_R\right)(\ketbra{\phi}{\phi}_{AR})} \geq 1-\varepsilon^2\,.
\end{align}
We also have for $\sigma_B\in\cS(B)$ that
\begin{align}
    \Tr{\ketbra{\phi}{\phi}_{BR}\left(\sigma_B\otimes\frac{I_R}{|R|}\right)} = \frac{1}{|R|^2}\,.
\end{align}
Let $\cM_{A'\rightarrow B'}$ be a channel that outputs a fixed state $\sigma_{B'}$ for any input and $\tilde{\cM}_{A\rightarrow B}$ be defined with the encoding operation, decoding operation and resource state chosen to be the same as those of $\tilde{\cN}$. It then holds that
\begin{align}
    &D^{\varepsilon^2}_h\left(\left(\tilde{\cN}\otimes\cI_R\right)(\ketbra{\phi}{\phi})\middle\|\left(\tilde{\cM}\otimes\cI_R\right)(\ketbra{\phi}{\phi})\right)\nonumber \\ 
    &= \max_{\substack{0\leq T \leq I^{d^2}\\ \Tr{T\left(\tilde{\cN}\otimes\cI_R\right)(\ketbra{\phi}{\phi})}\geq 1-\varepsilon^2}} -\log \Tr{T\left(\tilde{\cM}\otimes\cI_R\right)(\ketbra{\phi}{\phi})} \\
    &\geq -\log \Tr{\ketbra{\phi}{\phi}\left(\tilde{\cM}\otimes\cI_R\right)(\ketbra{\phi}{\phi})} \\
    &= 2\log |R|\,.
\end{align}
Let the set of all channels $\cM_{A'\rightarrow B'}$ with constant output be denoted by $\cM_{\text{const}}$. Then we have
\begin{align}
    2\log |R| &\leq \min\limits_{\cM\in\cM_{\text{const}}}D^{\varepsilon^2}_h\left(\left(\tilde{\cN}\otimes\cI_R\right)(\ketbra{\phi}{\phi})\middle\|\left(\tilde{\cM}\otimes\cI_R\right)(\ketbra{\phi}{\phi})\right) \\
    &\leq \max_{\omega_{AR}}\min\limits_{\cM\in\cM_{\text{const}}}D^{\varepsilon^2}_h\left(\left(\tilde{\cN}\otimes\cI_R\right)(\omega_{AR})\middle\|\left(\tilde{\cM}\otimes\cI_R\right)(\omega_{AR})\right) \\
    &\leq \max_{\rho_{A'R}}\min\limits_{\cM\in\cM_{\text{const}}}D^{\varepsilon^2}_h\left(\left(\cN\otimes\cI_R\right)(\rho_{A'R})\middle\|\left(\cM\otimes\cI_R\right)(\rho_{A'R})\right) \\
    &\leq \max_{\rho_{A'R}}\min_{\sigma_{B'}}D^{\varepsilon^2}_h\left(\left(\cN\otimes\cI_R\right)(\rho_{A'R})\middle\|\sigma_{B'}\otimes\rho_R\right)\,,
\end{align}
where the third inequality follows by data processing inequalities due to removing the encoding and decoding operations and the last inequality holds since $\cM\in\cM_{\text{const}}$.
\end{proof}
\vspace{5pt}

For quantum channels that are covariant with respect to unitaries on the input quantum register (see~\cite{holevo2002remarks} or \cite{datta2016second} for a more precise definition), the maximization over input states is achieved by the $n$-fold tensor product of maximally entangled states~\cite[Eq (4.81)]{datta2016second}. For such covariant channels, we have the following
\begin{align}
    r^\star_{\varepsilon}(\cN^{\otimes n}) &\leq \min_{\sigma_{{B'}^n}}\frac{1}{2}D^{\varepsilon^2}_h\left(\left(\left(\cN\otimes\cI_R\right)(\ketbra{\phi}{\phi}_{A'R})\right)^{\otimes n}\middle\|\sigma_{{B'}^n}\otimes \left(\frac{I_{R}}{d}\right)^{\otimes n}\right) \\
    &\leq \frac{1}{2}D^{\varepsilon^2}_h\left(\left(\left(\cN\otimes\cI_R\right)(\ketbra{\phi}{\phi}_{A'R})\right)^{\otimes n}\middle\| (\cN(\rho_{A'}))^{\otimes n}\otimes\left(\frac{I_{R}}{d}\right)^{\otimes n}\right)\,.
\end{align}
For the particular case where the error is chosen to be $1-\varepsilon_n = 1-e^{-na_n^2}$ and using the AEP for the hypothesis testing relative entropy (Proposition \ref{prop:chubbs_dh_expansion}), we get that for $\eta>0$, there exists $n^\star\in\mathbb{N}$ such that for $n\geq n^\star$ we have
\begin{align}
    \frac{1}{n}r^\star_{1-\varepsilon}(\cN^{\otimes n}) &\leq  \frac{1}{2n} D^{(1 - 2\varepsilon_n +\varepsilon_n^2)}_h\left(\left(\left(\cN\otimes\cI_R\right)(\ketbra{\phi}{\phi}_{A'R})\right)^{\otimes n}\middle\| (\cN(\rho_{A'}))^{\otimes n}\otimes\left(\frac{I_{R}}{d}\right)^{\otimes n}\right) \\
    &\leq \frac{1}{2}\left(I(R:B)_{(\cN\otimes I)(\ketbra{\phi}{\phi})} + \sqrt{2V(R:B)_{(\cN\otimes I)(\ketbra{\phi}{\phi})}}a_n + \eta a_n\right) \\
    &\leq C(\cN) + \frac{a_n}{\sqrt{2}}\sqrt{V_{\max}(B:R)_{\cN}} + o(a_n)\,,
\end{align}
which matches the converse statement for channel coding in the moderate deviation regime (Proposition \ref{prop:converse_result_channel_coding_mod_dev}).


\end{document}